\documentclass[prodmode,acmjacm]{myacmsmall} 
\usepackage{graphicx}
\usepackage{mathrsfs}
\usepackage{subfigure}
\usepackage{epsfig}
\usepackage{url}
\usepackage{latexsym}
\usepackage{enumerate}
\usepackage{amsmath}
\usepackage{multirow}
\usepackage{algorithm,algpseudocode}
\usepackage[sort,nocompress]{cite}

\usepackage{tikz}
\usetikzlibrary{positioning,arrows,calc}
\usepackage{pgfplots}

\newtheorem{problem}[theorem]{Problem}

\def\drop#1{}

\renewcommand{\vec}[1]{\mathbf{#1}}

\newcommand{\E}{\mathbb{E}}
\renewcommand{\P}{\mathbb{P}}

\newcommand{\In}{{\delta}}

\numberwithin{equation}{section}
\numberwithin{figure}{section}

\makeatletter
\newcounter{todo}

\def\l@todo#1#2{\par\noindent #1\hfill #2}
\def\listoftodos{\section*{TODOs} \@starttoc{loc} \medskip \hrule \medskip}
\makeatother

\title{Efficient SimRank Computation via Linearization\footnote{The manuscript with the same title had been available for a while, but this version extends the contexts significantly. Indeed this paper combines a journal version of our papers appeared in SIGMOD'14~\cite{kusumoto2014scalable} and ICDE'15~\cite{maehara2015scalable} with the previous manuscript, and in addition, we add significantly more details for mathematical analysis.}}

\author{
Takanori Maehara\thanks{maehara@nii.ac.jp}$^\dagger$
\affil{National Institute of Informatics}
Mitsuru Kusumoto\thanks{mkusumoto@preferred.jp}$^\dagger$
\affil{Preferred Infrastructure, Inc.}
Ken-ichi Kawarabayashi\thanks{k\_keniti@nii.ac.jp}$^\dagger$
\affil{National Institute of Informatics}
}

\begin{document}
\maketitle
\renewcommand{\thefootnote}{\fnsymbol{footnote}}
\footnotetext[2]{JST, ERATO, Kawarabayashi Large Graph Project}
\footnotetext[3]{Supported by JST, ERATO, Kawarabayashi Large Graph Project}
\renewcommand{\thefootnote}{\arabic{footnote}}

SimRank, proposed by Jeh and Widom, provides a good similarity measure that has been successfully used in numerous applications.
While there are many algorithms proposed for computing SimRank,
their computational costs are very high.

In this paper, we propose a new computational technique, ``SimRank linearization,'' for computing SimRank, which converts the SimRank problem to a linear equation problem.
By using this technique, we can solve many SimRank problems, such as single-pair compuation, single-source computation, all-pairs computation, top $k$ searching, and similarity join problems, efficiently.

\medskip


\section{Introduction}
\label{sec:introduction}

\subsection{Background and motivation}

Very large-scale networks are ubiquitous in today's world,
and designing scalable algorithms for such huge network has become a pertinent problem in all aspects of compute science.
The primary problem is the vast size of modern graph datasets.
For example, the World Wide Web currently consists of over one trillion links and is expected to exceed tens of trillions in the near future,
and Facebook embraces over 800 million active users, with hundreds of billions of friend links.

Large graphs arise in numerous applications where
both the basic entities and the relationships between these entities are given.
A graph stores the \emph{objects} of the data on its vertices,
and represents the \emph{relations} among these objects by its edges.
For example,
the vertices and edges of the World Wide Web graph
correspond to the webpages and hyperlinks, respectively.
Another typical graph is a social network,
whose vertices and edges correspond to personal information
and friendship relations, respectively.

With the rapidly increasing amount of graph data, 
the \emph{similarity search problem}, which identifies similar vertices in a graph,
has become an important problem
with many applications, including 
web analysis~\cite{jeh2002simrank,liben2007link},
graph clustering~\cite{yin2006linkclus,zhou2009graph},
spam detection~\cite{gyongyi2004combating},
computational advertisement~\cite{antonellis2008simrank},
recommender systems~\cite{abbassi2007recommender,yu2010simrate},
and 
natural language processing~\cite{sheible2010sentiment}.

Several similarity measures have been proposed.
For example, bibliographic coupling~\cite{kessler1963bibliographic}, co-citation~\cite{smart1973cocitation}, P-Rank~\cite{zhao2009prank},
PageSim~\cite{lin2006pagesim}, Extended Nearest Neighborhood Structure~\cite{lin2007extending}, MatchSim~\cite{lin2012matchsim}, and so on.
In this paper, we consider \emph{SimRank}, a link-based similarity measure proposed by Jeh and Widom~\cite{jeh2002simrank} for searching web pages.
SimRank supposes that ``two similar pages are linked from many similar pages.''
This intuitive concept is formulated by the following recursive definition:
For a graph $G = (V, E)$, the SimRank score $s(i,j)$ of a pair of vertices $(i, j) \in V \times V$
is recursively defined by
\begin{align} 
\label{eq:simrankoriginal}
s(i,j) := \begin{cases} 1, & i = j, \\ \displaystyle  \frac{c}{|\In(i)| |\In(j)|} \sum_{i' \in \In(i), j \in \In(j)} s(i', j'), & i \neq j, \end{cases}
\end{align}
where $\In(i) = \{ j \in V : (j,i) \in E\}$ is the set of in-neighbors of $i$,
and $c \in (0,1) $ is a decay factor usually set to $c = 0.8$~\cite{jeh2002simrank} or $c = 0.6$~\cite{lizorkin2010accuracy}.
See Figure~\ref{fig:example} for an example of the SimRank on a small graph.

SimRank can be regarded as a \emph{label propagation}~\cite{zhu2002learning} on the squared graph.
Let us consider a squared graph $G^2 = (V^{(2)}, E^{(2)})$ whose vertices are the pair of vertices $V^{(2)} = V \times V$
and edges are defined by
\begin{align}
  E^{(2)} = \{ ((i,j), (i',j')) : (i,i'), (j,j') \in E \}.
\end{align}
Then the SimRank is a label propagation method with trivial relations $s(i,i) = 1$ (i.e., $i$ is similar to $i$) for all $i \in V$ on $G^{(2)}$.

SimRank also has a ``random-walk'' interpretation.
Let us consider two random walks that start from vertices $i$ and $j$, respectively,
and follow the in-links. Let $i^{(t)}$ and $j^{(t)}$ be the $t$-th position of each random walk, respectively.
The first meeting time $\tau_{i,j}$ is defined by
\begin{align} \label{eq:firstmeetingtime}
  \tau_{i j} = \min \{ t : i^{(t)} = j^{(t)} \}.
\end{align}
Then SimRank score is obtained by
\begin{align} \label{eq:randomsurferpair}
s(i,j) = \E[ c^{\tau_{i,j}} ].
\end{align}

\begin{figure}
\centering
\begin{minipage}{0.30\hsize}
\centering
\input{fig_simrank}
\end{minipage}
\begin{minipage}{0.68\hsize}
\centering
\footnotesize
\tabcolsep=2pt
\begin{tabular}{ccc} \hline
$i$ & $j$ & $s(i,j)$ \\ \hline
1 & 2 & 0.260 \\
1 & 3 & 0.142 \\
1 & 4 & 0.120 \\
1 & 5 & 0.162 \\
1 & 6 & 0.069 \\
1 & 7 & 0.219 \\
2 & 3 & 0.121 \\
\hline
\end{tabular} \quad
\begin{tabular}{ccc} \hline
$i$ & $j$ & $s(i,j)$ \\ \hline
2 & 4 & 0.141 \\
2 & 5 & 0.132 \\
2 & 6 & 0.069 \\
2 & 7 & 0.226 \\
3 & 4 & 0.128 \\
3 & 5 & 0.230 \\
3 & 6 & 0.236 \\
\hline
\end{tabular} \quad
\begin{tabular}{ccc} \hline
$i$ & $j$ & $s(i,j)$ \\ \hline
3 & 7 & 0.101 \\
4 & 5 & 0.107 \\
4 & 6 & 0.080 \\
4 & 7 & 0.125 \\
5 & 6 & 0.271 \\
5 & 7 & 0.110 \\
6 & 7 & 0.061 \\
\hline
\end{tabular}

\end{minipage}
\caption{Example of the SimRank. $c = 0.6$.}
\label{fig:example}
\end{figure}

SimRank and its related measures (e.g., SimRank++~\cite{antonellis2008simrank},  S-SimRank~\cite{cai2008ssimrank}, P-Rank~\cite{zhao2009prank}, and SimRank$^*$~\cite{yu2013more})
give high-quality scores in activities such as natural language processing~\cite{sheible2010sentiment}, 
computational advertisement~\cite{antonellis2008simrank}, 
collaborative filtering~\cite{yu2010simrate}, and
web analysis~\cite{jeh2002simrank}.
As implied in its definition,
SimRank exploits the information in multihop neighborhoods.
In contrast, most other similarity measures utilize only the one-step neighborhoods.
Consequently, SimRank is more effective than other similarity measures in real applications.

Although SimRank is naturally defined and gives high-quality similarity measure,
it is not so widely used in practice, due to high computational cost.
While there are several algorithms proposed so far to compute SimRank scores,
unfortunately, their computation costs (in both time and space) are very expensive.
The difficulty of computing SimRank may be viewed as follows:
to compute a SimRank score $s(u,v)$ for two vertices $u,v$,
since \eqref{eq:simrankoriginal} is defined recursively,
we have to compute SimRank scores for all $O(n^2)$ pairs of vertices.
Therefore
it requires $O(n^2)$ space and $O(n^2)$ time, where $n$ is the number of vertices.
In order to reduce this computation cost, several approaches have been proposed.
We review these approaches in the following subsection.

\subsection{Related Work}
\label{sec:relatedwork}

In order to reduce this computation cost, several approaches have been proposed~\cite{fogaras2005scaling,he2010parallel,li2010fastcomp,li2010fast,lizorkin2010accuracy,yu2010taming,yu2013more,yu2012space}.
Here, we briefly survey some existing computational techniques for SimRank.
We summarize the existing results in Table~\ref{tbl:Complexity}.
Let us point out that there are three fundamental problems for SimRank:
(1) single-pair SimRank to compute $s(u,v)$ for given two vertices $u$ and $v$,
(2) single-source SimRank to compute $s(u,v)$ for a given vertex $u$ and all other vertices $v$,
and (3) all-pairs SimRank to compute $s(u,v)$ for all pair of vertices $u$ and $v$.

In the original paper by Jeh and Widom~\cite{jeh2002simrank},
all-pairs SimRank scores are computed by recursively evaluating the
equation \eqref{eq:simrankoriginal} for all $u, v \in V$.
This ``naive'' computation yields an $O(T d^2 n^2)$ time algorithm, where $T$ denotes the number of iterations
and $d$ denotes the average degree of a given network.
Lizorkin et al.~\cite{lizorkin2010accuracy} proposed a ``partial sum'' technique,
which memorizes partial calculations of Jeh and Widom's algorithm
to reduce the time complexity of their algorithm.
This leads to an $O(T \min \{n m, n^3 / \log n \} )$ algorithm.
Yu et al.~\cite{yu2012space} applied the fast matrix multiplication~\cite{strassen1969gaussian,williams2012multiplying}
and then obtained an $O(T \min \{n m, n^\omega \} )$ algorithm to compute
all pairs SimRank scores, where $\omega < 2.373$ is the exponent of matrix multiplication.
Note that the space complexity of these algorithms is $O(n^2)$,
since they have to maintain all SimRank scores for each pair of vertices
to evaluate the equation \eqref{eq:simrankoriginal}.
This results is, so far, the state-of-the-art algorithm to compute
SimRank scores for all pairs of vertices.

There are some algorithms based on a random-walk interpretation~\eqref{eq:randomsurferpair}.
Fogaras and R\'acz~\cite{fogaras2005scaling} evaluate the right-hand side by
Monte-Carlo simulation with a fingerprint tree data structure, and
they obtained a faster algorithm to compute single pair SimRank score for
given two vertices $i, j$.
Li et al.~\cite{li2010fast} also proposed an algorithm based on the random-walk iterpretation;
however their algorithm is an iterative algorithm to compute the first meeting time
and computes all-pairs SimRank deterministically.

Some papers proposed spectral decomposition based algorithms
(e.g., \cite{fujiwara2013efficient,he2010parallel,li2010fastcomp,yu2010taming,yu2013more}), but there is a mistake
in the formulation of SimRank. On the other hand, their algorithms may output reasonable results.
We shall mention more details about these algorithms in Remark~\ref{rem:wrong}.


%

\subsection{Contribution}

In this paper, we propose a novel computational technique for SimRank, called \emph{SimRank linearlization}.
This technique allows us to solve many kinds of SimRank problems such as the following:
\begin{description}
  \item[Single-pair SimRank] We are given two vertices $i, j \in V$, compute SimRank score $s(i,j)$.
  \item[Single-source SimRank] We are given a vertex $i \in V$, compute SimRank scores $s(i,j)$ for all $j \in V$.
  \item[All-pairs SimRank] Compute SimRank scores $s(i,j)$ for all $i, j \in V$.
  \item[Top $k$ SimRank search] We are given a vertex $i \in V$, return $k$ vertices $j$ with $k$ highest SimRank scores.
  \item[SimRank join] We are given a threshold $\delta$, return all pairs $(i,j)$ such that $s(i,j) \ge \delta$.
\end{description}
%
For all problems, the proposed algorithm outperforms the existing methods.

\subsection{Organization}

The paper consists of four parts.
In Section~\ref{sec:linearization}, 
we introduce the SimRank linearization technique
and show that how to use the linearization to solve single-pair, single-source, and all-pairs problem.
In Section~\ref{sec:topk},
we describe how to solve top $k$ SimRank search problem.
In Section~\ref{sec:join},
we describe how to solve SimRank join problem.
Each section contains computational experiments.

\begin{table}
\tbl{Complexity of SimRank algorithms. $n$ denotes the number of vertices, $m$ denotes the number of edges, $d$ denotes the average degree, $T$ denotes the number of iterations, $R$ is the number of Monte-Carlo samples, and $r$ denotes the rank for low-rank approximation. Note that $*$-marked method is based on an incorect formula: see Remark~\ref{rem:wrong}. \label{tbl:Complexity}}{
\centering
{ \small
\begin{tabular}{l|llll}
Algorithm & Type & Time & Space & Technique \\ \hline
Proposed (Section~\ref{sec:algorithm}) & Single-pair & $O(T m)$ & $O(m)$ & Linearization \\
Proposed (Section~\ref{sec:algorithm}) & Single-source & $O(T^2 m)$ & $O(m)$ & Linearization \\
Proposed (Section~\ref{sec:algorithm}) & All-pairs & $O(T^2 n m)$ & $O(m)$ & Linearization \\
Proposed (Section~\ref{sec:algorithm}) & Top-$k$ search & $\ll O(T R)$ & $O(m)$ & Linearization \& Monte Carlo \\
Proposed (Section~\ref{sec:algorithm}) & Join & $\approx O(\text{output})$ & $O(m+\text{output})$ & Linearization \& Gauss-Southwell \\ \hline
\cite{li2010fast} & Single-pair & $O(T d^2 n^2)$ & $O(n^2)$ & Random surfer pair (Iterative) \\
\cite{fogaras2005scaling} & Single-pair & $O(T R)$ & $O(m+n R)$ & Random surfer pair (Monte Carlo)\\
\cite{jeh2002simrank} & All-pairs & $O(T n^2 d^2)$ & $O(n^2)$ & Naive \\
\cite{lizorkin2010accuracy} & All-pairs & $O(T \min\{n m, n^3 /\log n\})$ & $O(n^2)$ & Partial sum \\
\cite{yu2012space} & All-pairs & $O(\min\{n m, n^\omega\})$ & $O(n^2)$ & Fast matrix multiplication \\
\cite{li2009exploiting} & All-pairs & $O(^{4/3})$ & $O(n^{4/3})$ & Block partition \\ 
\cite{li2010fastcomp} & All-pairs & $O(r^4 n^2)$ & $O(n^2)$ & Singular value decomposition$^*$ \\
\cite{fujiwara2013efficient} & All-pairs & $O(r^4 n)$ & $O(r^2 n^2)$ & Singular value decomposition$^*$ \\
\cite{yu2010taming} & All-pairs & $O(n^3)$ & $O(n^2)$ & Eigenvalue decomposition$^*$ \\
\end{tabular}}
}
\end{table}

\begin{table}[tb]
\tbl{List of symbols \label{tbl:symbols}}{
\begin{tabular}{cl} \hline
symbol & description \\ \hline
$G$ & directed unweighted graph, $G = (V, E)$ \\
$V$ & set of vertices \\
$E$ & set of edges \\
$n$ & number of vertices, $n = |V|$ \\
$m$ & number of edges, $m = |E|$ \\
$i,j$ & vertex \\
$e$ & edge \\
$\delta(i)$ & in-neighbors of of $i$, $\delta(i) = \{ j \in V : (j,i) \in E \}$ \\
$P$ & transition matrix, $P_{ij} = 1/|\delta(j)|$ for $(i,j) \in E$ \\
$s(i,j)$ & SimRank of $i$ and $j$ \\
$S$ & SimRank matrix, $S_{ij} = s(i,j)$ \\
$D$ & diagonal correction matrix, $S = c P^\top S P + D$ \\
\hline
\end{tabular}}
\end{table}

%

\section{Linearized SimRank}
\label{sec:linearization}

\subsection{Concept of linearized SimRank}

Let us first observe the difficulty in computing SimRank.
Let $G = (V, E)$ be a directed graph, and
let $P = (P_{ij})$ be a transition matrix of transpose graph $G^\top$ defined by
\begin{align*}
P_{ij} := \begin{cases}
1 / |\In(j)|, & (i,j) \in E, \\
0, & (i,j) \not \in E,
\end{cases}
\end{align*}
where $\In(i) = \{ j \in V : (j, i) \in E \}$ denotes the in-neighbors of $i \in V$.
Let $S = (s(i,j))$ be the \emph{SimRank matrix}, whose $(i,j)$ entry is the SimRank score of $i$ and $j$.
Then the SimRank equation~\eqref{eq:simrankoriginal} is represented~\cite{yu2012space} by:
\begin{align} \label{eq:simrank}
  S = (c P^\top S P) \lor I,
\end{align}
where $I$ is the identity matrix, and $\lor$ denotes the element-wise maximum, i.e.,
$(i,j)$ entry of the matrix $A \lor B$ is given by $\max \{ A_{ij}, B_{ij} \}$.

In our view, the difficulty in computing SimRank via equation \eqref{eq:simrank}
comes from the element-wise maximum, which is a \emph{non-linear} operation.
To avoid the element-wise maximum, we introduce a new formulation of SimRank as follows.
By observing \eqref{eq:simrank}, since $S$ and $c P^\top S P$ only differ in their diagonal elements,
there exists a diagonal matrix $D$ such that
\begin{align} \label{eq:linearizedsimrank}
  S = c P^\top S P + D.
\end{align}
We call such a matrix $D$ the \emph{diagonal correction matrix}.
The main idea of our approach here is to split a SimRank problem into the following two subproblems:
\begin{enumerate}
  \item Estimate diagonal correction matrix $D$.
  \item Solve the SimRank problem using $D$ and the linear recurrence equation \eqref{eq:linearizedsimrank}.
\end{enumerate}
For efficient computation, we must estimate $D$ without computing the whole part of $S$.

To simplify the discussion, we introduce the notion of \emph{linearized SimRank}.
Let $\Theta$ be an $n \times n$ matrix.
A linearized SimRank $S^L(\Theta)$ is a matrix that satisfies
the following \emph{linear} recurrence equation:
\begin{align} \label{eq:LinearizedSimRank}
  S^L(\Theta) = c P^\top S^L(\Theta) P + \Theta.
\end{align}
Below, we provide an example that illustrates what linearized SimRank is.
\begin{example}[Star graph of order $4$] \label{ex:LinearizedSimRank}
Let $G$ be a star graph of order $4$ (i.e., $G$ has one vertex of degree three and three vertices of degree one).
The transition matrix (of the transposed graph) is
\[
P = \begin{bmatrix}
  0 & 1 & 1 & 1 \\
  1/3 & 0 & 0 & 0 \\
  1/3 & 0 & 0 & 0 \\
  1/3 & 0 & 0 & 0 \\
\end{bmatrix},
\]
and SimRank for $c = 0.8$ is
\[
S = \begin{bmatrix}
  1 & 0 & 0 & 0 \\
  0 & 1 & 4/5 & 4/5 \\
  0 & 4/5 & 1 & 4/5 \\
  0 & 4/5 & 4/5 & 1 \\
\end{bmatrix}.
\]
Thus, the diagonal correction matrix $D$ is obtained by
\[
D = S - c P^\top S P = \mathrm{diag}(23/75, 1/5, 1/5, 1/5),
\]
\end{example}

\begin{remark}
\label{rem:wrong}
Some papers have used the following formula for SimRank
(e.g.,
equation (2) in \cite{fujiwara2013efficient},
equation (2) in \cite{he2010parallel},
equation (2) in \cite{li2010fastcomp},
and
equation (3) in \cite{yu2013more}%
):
\begin{align} \label{eq:WrongSimRank}
  S = c P^\top S P + (1 - c) I.
\end{align}
However, this formula does not hold;
\eqref{eq:WrongSimRank} requires diagonal correction matrix $D$ to have the same diagonal entries,
but Example~\ref{ex:LinearizedSimRank} is a counterexample.
In fact, matrix $S$ defined by \eqref{eq:WrongSimRank}
is a linearized SimRank $S^L(\Theta)$ for a matrix $\Theta = (1 - c) I$.
\end{remark}
We provide some basic properties of linearized SimRank in Appendix.

%

\subsection{Solving SimRank problems via linearization}
\label{sec:LinearizedSimRankComputation}

In this section,
we present our proposed algorithms for SimRank
by assuming that the diagonal correction matrix $D$ has already been obtained.
All algorithms are based on the same fundamental idea; i.e.,
in \eqref{eq:linearizedsimrank},
by recursively substituting the left hand side into the right hand side,
we obtain the following series expansion:
\begin{align}
\label{eq:NeumannExpansion}
  S = D + c P^\top D P + c^2 P^{\top 2} D P^2 + \cdots.
\end{align}
Our algorithms compute SimRank by evaluating
the first $T$ terms of the above series.
The time complexity of the algorithms are
$O(T m)$ for the single-pair problem,
$O(T^2 m)$ for the single-source problem,
and $O(T^2 n m)$ for the all-pairs problem.
For all problems, the space complexity is $O(m)$.

\subsubsection{Single-pair SimRank}

Let $e_i$ be the $i$-th unit vector ($i = 1, \ldots, n$);
then SimRank score $s(i,j)$ is obtained via the $(i,j)$ component
of SimRank matrix $S$, i.e., $s(i,j) = e_i^\top S e_j$.
Thus, by applying $e_i^\top$ and $e_j$ to both sides of \eqref{eq:NeumannExpansion}, we obtain
\begin{align}
\label{eq:ForSinglePair}
  e_i^\top S e_j &= e_i^\top D e_j + c (P e_i)^\top D P e_j \nonumber \\ & + c^2 (P^2 e_i)^\top D P^2 e_j + \cdots.
\end{align}
Our single-pair algorithm (Algorithm~\ref{alg:OneVsOne}) evaluates the right-hand side of \eqref{eq:ForSinglePair}
by maintaining $P^t e_i$ and $P^t e_j$.
The time complexity is $O(T m)$ since the algorithm performs $O(T)$ matrix vector products
for $P^t e_i$ and $P^t e_j$ ($t = 1, \ldots, T-1$).

\begin{algorithm}[tb]
\caption{Single-pair SimRank} \label{alg:OneVsOne}
\begin{algorithmic}[1]
\Procedure{SinglePairSimRank}{i,j}
\State{$\alpha \leftarrow 0$, $x \leftarrow e_i$, $y \leftarrow e_j$}
\For{$t = 0, 1, \ldots, T-1$}
\State{$\alpha \leftarrow \alpha + c^t x^\top D y$, $x \leftarrow P x$, $y \leftarrow P y$}
\EndFor
\State{Report $S_{ij} = \alpha$}
\EndProcedure
\end{algorithmic}
\end{algorithm}

\subsubsection{Single-source SimRank}

For the single-source problem, to obtain $s(i,j)$ for all $j \in V$,
we need only compute vector $S e_i$,
because its $j$-th component is $s(i,j)$.
By applying $e_i$ to \eqref{eq:NeumannExpansion}, we obtain
\begin{align}
\label{eq:ForSingleSource}
  S e_i = D e_i + c P^\top D P e_i + c^2 P^{\top 2} D P^2 e_i + \cdots.
\end{align}
Our single-source algorithm (Algorithm~\ref{alg:OneVsAll}) evaluates the right hand side of \eqref{eq:ForSingleSource}
by maintaining $P^t e_i$ and $P^t e_j$.
The time complexity is $O(T^2 m)$ since it performs $O(T^2)$ matrix vector products
for $P^{\top t} D P^t e_i$ ($t = 1, \ldots, T-1$).

\begin{algorithm}[tb]
\caption{Single-source SimRank} \label{alg:OneVsAll}
\begin{algorithmic}[1]
\Procedure{SingleSourceSimRank}{i}
\State{$\gamma \leftarrow \vec{0}$, $x \leftarrow e_i$}
\For{$t = 0, 1, \ldots, T-1$}
\State{$\gamma \leftarrow \gamma + c^t P^{\top t} D x$, $x \leftarrow P x$}
\EndFor
\State{Report $S_{ij} = \gamma_j$ for $j = 1, \ldots, n$}
\EndProcedure
\end{algorithmic}
\end{algorithm}

Note that, if we can use an additional $O(T n)$ space, the single-source problem can be solved in $O(T m)$ time.
We first compute $u_t = D P^t e_i$ for all $t = 1, \ldots, T$, and store them;
this requires $O(T m)$ time and $O(T n)$ additional space.
Then, we have
\begin{align}
  S e_i = u_0 + c P^\top ( u_1 + \cdots (c P^{\top} (u_{t-1} + c P^{\top} u_T)) \cdots ),
\end{align}
which can be computed in $O(T m)$ time. 
We do not use this technique in our experiment because we assume that the network is very large and hence $O(T n)$ is expensive.

%
%
%

\subsubsection{All-pairs SimRank}
\label{sec:allpairs}

Computing all-pairs SimRank is an expensive task for a large network,
because it requires $O(n^2)$ time since the number of pairs is $n^2$.
To compute all-pairs SimRank, it is best to avoid using $O(n^2)$ space.

Our all-pairs SimRank algorithm
applies the single-source SimRank algorithm (Algorithm~\ref{alg:OneVsAll}) for all initial vertices,
as shown in Algorithm~\ref{alg:AllVsAll}.
The complexity is $O(T^2 n m)$ time and requires only $O(m)$ space.
Since the best-known all-pairs SimRank algorithm~\cite{lizorkin2010accuracy} requires $O(T n m)$ time and $O(n^2)$ space,
our algorithm significantly improves the space complexity
and has almost the same time complexity
(since the cost of factor $T$ is much smaller than $n$ or $m$).

It is worth noting that this algorithm is distributed computing friendly.
If we have $M$ machines,
we assign initial vertices to each machine and independently compute the single-source SimRank.
Then the computational time is reduced to $O(T^2 n m / M)$.
This shows the scalability of our all-pairs algorithm.
\begin{algorithm}[tb]
\caption{All-pairs SimRank} \label{alg:AllVsAll}
\begin{algorithmic}[1]
\Procedure{AllPairsSimRank}{}
\For{$i = 1, \ldots, n$}
\State{Compute SingleSourceSimRank($i$)}
\EndFor
\EndProcedure
\end{algorithmic}
\end{algorithm}

\subsection{Diagonal correction matrix estimation}
\label{sec:DiagonalEstimation}


We first observe that the diagonal correction matrix is
uniquely determined from the diagonal condition.
\begin{proposition}
\label{prop:DiagonalCondition}
A diagonal matrix $D$ is the diagonal correction matrix, i.e., $S^L(D) = S$
if and only if $D$ satisfies
\begin{align}
\label{eq:DiagonalCondition}
  S^L(D)_{kk} = 1, \quad (k = 1, \ldots, n),
\end{align}
where $S^L(D)_{kk}$ denotes $(k,k)$ entry of the linearized SimRank matrix $S^L(D)$.
\end{proposition}
\begin{proof}
See Appendix.
\end{proof}

This proposition shows that 
the diagonal correction matrix can be estimated 
by solving equation \eqref{eq:DiagonalCondition}.
Furthermore, we observe that,
since $S^L$ is a linear operator,
\eqref{eq:DiagonalCondition} is a \emph{linear equation}
with $n$ real variables $D_{11}, \ldots, D_{nn}$
where $D = \mathrm{diag}(D_{11}, \ldots, D_{nn})$.
Therefore, we can apply a numerical linear algebraic method 
to estimate matrix $D$.

The problem for solving \eqref{eq:DiagonalCondition} lies in the complexity.
To reduce the complexity,
we combine an \emph{alternating method} (a.k.a. the \emph{Gauss-Seidel method})
with \emph{Monte Carlo simulation}. 
The complexity of the obtained algorithm is $O(T L R n)$ time, 
where $L$ is the number of iterations for the alternating method,
and $R$ is the number of Monte Carlo samples.
We analyze the upper bound of parameters $L$ and $R$ for sufficient accuracy
in Subsection~\ref{sec:accuracy} below.

\subsubsection{Alternating method for diagonal estimation}
\label{sec:Alternating}

Our algorithm is motivated by the following intuition:
\begin{align}
  \label{eq:Intuition}
  \begin{tabular}{l}
  A $(k,k)$ diagonal entry $S^L(D)_{kk}$ is the most \\
  affected by the $(k,k)$ diagonal entry $D_{kk}$ of $D$.
  \end{tabular}
\end{align}
This intuition leads to the following iterative algorithm.
Let $D$ be an initial guess\footnote{We discuss an initial solution in Remark~\ref{rem:initial} in Appendix.};
for each $k = 1, \ldots, n$, the algorithm iteratively updates $D_{kk}$ to satisfy $S^L(D)_{kk} = 1$.
The update is performed as follows. 
Let $E^{(k,k)}$ be the matrix whose $(k,k)$ entry is one, with the other entries being zero.
To update $D_{kk}$, we must find $\delta \in \mathbb{R}$ such that
\[
  S^L(D + \delta E^{(k,k)})_{kk} = 1.
\]
Since $S^L$ is linear, the above equation is solved as follows:
\begin{align}
\label{eq:ClosedForm}
  \delta = \frac{1 - S^L(D)_{kk}}{S^L(E^{(k,k)})_{kk}}.
\end{align}
This algorithm is shown in Algorithm~\ref{alg:DiagonalEstimation}.

Mathematically, the intuition \eqref{eq:Intuition} shows the 
\emph{diagonally dominant} property of operator $S^L$.
Furthermore, the obtained algorithm (i.e., Algorithm~\ref{alg:DiagonalEstimation}) is the \emph{Gauss-Seidel} method for a linear equation.
Since the Gauss-Seidel method converges for a diagonally dominant operator~\cite{golub2012matrix},
Algorithm~\ref{alg:DiagonalEstimation} converges to
the diagonal correction matrix%
\footnote{Strictly speaking, we need some conditions for the diagonally dominant property of operator $S^L$. In practice, we can expect the estimation algorithm converges; see Lemma~\ref{lem:DiagonallyDominant} in Appendix.}
.

\begin{algorithm}[tb]
\caption{Diagonal estimation algorithm.} \label{alg:DiagonalEstimation}
\begin{algorithmic}[1]
\Procedure{DiagonalEstimation}{}
\State{Set initial guess of $D$.}
\For{$\ell = 1, \ldots, L$}
\For{$k = 1, \ldots, n$}
\State{$\delta \leftarrow (1 - S^L(D)_{kk}) / S^L(E^{(k,k)})_{kk}$}
\State{$D_{kk} \leftarrow D_{kk} + \delta$}
\EndFor
\EndFor
\State{\Return $D$}
\EndProcedure
\end{algorithmic}
\end{algorithm}

\subsubsection{Monte Carlo based evaluation}
\label{monte}

For an efficient implementation of our diagonal estimation algorithm (Algorithm~\ref{alg:DiagonalEstimation}),
we must establish an efficient method to estimate $S^L(D)_{kk}$ and $S^L(E^{(k,k)})_{kk}$.

Consider a random walk that starts at vertex $k$ and follows its in-links.
Let $k^{(t)}$ denote the location of the random walk after $t$ steps.
Then we have
\[ \mathbf{E}[ e_{k^{(t)}} ] = P^t e_k. \]
We substitute this representation into \eqref{eq:ForSinglePair} 
and evaluate the expectation via Monte Carlo simulation.
Let $k^{(t)}_1, \ldots, k^{(t)}_R$ be $R$ independent random walks.
Then for each step $t$, we have estimation
\begin{align}
\label{eq:DefinitionPt}
  (P^t e_k)_i \approx \#\{ r = 1, \ldots, R : k^{(t)}_r = i \} / R =: p^{(t)}_{ki}.
\end{align}
Thus the $t$-th term of \eqref{eq:ForSinglePair} for $i = j = k$ is estimated as
\begin{align} \label{eq:expectedPDP}
  (P^t e_k)^\top D P^t e_k \approx \sum_{i=1}^n p^{(t) 2}_{ki} D_{ii}.
\end{align}
We therefore obtain 
Algorithm~\ref{alg:DiagonalAlphaBeta} for estimating 
$S^L(D)_{kk}$ and $S^L(E^{(k,k)})_{kk}$.

Using a hash table, we can implement Algorithm~\ref{alg:DiagonalAlphaBeta} in $O(T R)$ time, where $R$ denotes the number of samples and
$T$ denotes the maximum steps of random walks that are exactly the number of SimRank iterations.
Therefore Algorithm~\ref{alg:DiagonalEstimation} is performed in $O(T L R n)$ time,
where $L$ denotes the number of iterations required for Algorithm~\ref{alg:DiagonalEstimation}.

\begin{algorithm}[tb]
\caption{Estimate $S^L(D)_{kk}$ and $S^L(E^{(k,k)})_{kk}$.} \label{alg:DiagonalAlphaBeta}
\begin{algorithmic}[1]
\State{$\alpha \leftarrow 0$, $\beta \leftarrow 0$,
$k_1 \leftarrow k, k_2 \leftarrow k, \ldots, k_R \leftarrow k$}
\For{$t = 0, 1, \ldots, T-1$}
\For{$i \in \{ k_1, k_2, \ldots, k_R \}$}
\State{$p_{ki}^{(t)} \leftarrow \# \{ r = 1, \ldots, R : k_r = i \}/R$}
\If{$i = k$}
\State{$\alpha \leftarrow \alpha + c^t p_{ki}^{(t) 2}$}
\EndIf
\State{$\beta \leftarrow \beta + c^t p_{ki}^{(t) 2} D_{i i}$}
\EndFor
\For{$r = 1, \ldots, R$}
\State{$k_r \leftarrow \delta_-(k_r)$ randomly}
\EndFor
\EndFor
\State{\Return $S^L(D)_{kk} \approx \alpha$, $S^L(E^{(k,k)})_{kk} \approx \beta$.}
\end{algorithmic}
\end{algorithm}

\subsubsection{Correctness: Accuracy of the algorithm}
\label{sec:accuracy}

To complete the algorithm,
we provide a theoretical estimation of parameters $L$ and $R$
that are determined in relation to the desired accuracy.
In Section~\ref{sec:Experiments}, we experimentally evaluate the accuracy.

\medskip

\noindent \textbf{Estimation of the number of iterations $L$.}
\quad
The convergence rate of the Gauss-Seidel method is linear;
i.e., the squared error $ \sqrt{ \sum_{k=1}^n ( S^L(D)_{kk} - 1 )^2 } $
at $l$-th iteration of Algorithm~\ref{alg:DiagonalEstimation}
is estimated as $O( \rho^l )$, 
where $0 \le \rho < 1$ is a constant (i.e., the spectral radius of the iteration matrix).
Therefore, since the error of an initial solution is $O(n)$, 
the number of iterations $L$ of Algorithm~\ref{alg:DiagonalEstimation} is estimated as $O(\log (n / \epsilon))$
for desired accuracy $\epsilon$.

\medskip

\noindent \textbf{Estimation of the number of samples $R$.}
\quad
Since the algorithm is a Monte Carlo simulation,
there is a trade-off between accuracy and the number of samples $R$. 
The dependency is estimated by the Hoeffding inequality,
which is described below.
\begin{proposition}
\label{prop:Hoeffding}
Let $p^{(t)}_{ki}$ ($i = 1, \ldots, n$) be defined by \eqref{eq:DefinitionPt},
and let $p^{(t)}_k := (p^{(t)}_{k1}, \ldots, p^{(t)}_{kn})$.
Then
\begin{align*}
  \mathbf{P} \left\{ \| P^t e_k - p^{(t)}_k \| > \epsilon \right\} 
   \le 2 n \exp \left( - \frac{(1-c) R \epsilon^2}{2} \right).
\end{align*}
where $\mathbf{P}$ denotes the probability.
\end{proposition}
\begin{lemma} \label{lem:lem1}
Let $k_1^{(t)}, \ldots, k_R^{(t)}$ be positions of $t$-th step of independent random walks that start from 
a vertex $k$ and follow ln-links.
Let $X_k^{(t)} := (1/R) \sum_{r=1}^R e_{k_r^{(t)}}$.
Then for all $l = 1, \ldots, n$,
\begin{align*}
  P\left\{ \left| e_l^\top \left( X_k^{(t)} - P^t e_k \right) \right| \ge \epsilon \right\} \le 2 \exp \left( - 2 R \epsilon^2 \right).
\end{align*}
\end{lemma}
\begin{proof}
Since $\E[ e_{k_r^{(t)}} ] = P^t e_k$,
this is a direct application of the Hoeffding's inequality.
\end{proof}
\begin{proof}[of Proposition~\ref{prop:Hoeffding}]
Since $p_{ki}^{(t)}$ defined by \eqref{eq:DefinitionPt} satisfies $p_{ki}^{(t)} = e_i^\top X^{(t)}_k$,
by Lemma~\ref{lem:lem1}, we have
\begin{align*}
  P \left\{ \| P^t e_k - p^{(t)}_k \| > \epsilon \right\} 
  &\le
  n P \left\{ \left| e_i^\top P^t e_k - p^{(t)}_{ki} \right| > \epsilon \right\} \\
  &\le 
  2 n \exp \left( - 2 R \epsilon^2 \right). 
\end{align*}
\end{proof}
This shows that we need $R = O((\log n) / \epsilon^2)$ samples
to accurately estimate $P^t e_k$ via Monte Carlo simulation.

\medskip

By combining all estimations,
we conclude that diagonal correction matrix $D$ is estimated in
$O(T n \log (n/\epsilon) (\log n) / \epsilon^2)$ time.
Since this is nearly linear time, 
the algorithm scales well.

Note that the accuracy of our framework
only depends on the accuracy of the diagonal estimation,
i.e., if $D$ is accurately estimated, the SimRank matrix $S$ are accurately estimated by $S^L(D)$; see Proposition~\ref{prop:UniformBound} in Appendix.
Therefore, if we want accurate SimRank scores,
we only need to spend more time in the preprocessing phase
and fortunately do not need to increase the time required in the query phase.

\subsection{Experiments}
\label{sec:Experiments}

In this section, we evaluate our algorithm via experiments using real networks.
The datasets we used are shown in Table~\ref{tbl:dataset};
These are obtained from
``Stanford Large Network Dataset Collection\footnote{\url{http://snap.stanford.edu/data/index.html}},'',
``Laboratory for Web Algorithmics\footnote{\url{http://law.di.unimi.it/datasets.php}},'' and
``Social Computing Research\footnote{\url{http://socialnetworks.mpi-sws.org/datasets.html}}.''
We first evaluate the accuracy in Section~\ref{sec:Accuracy},
then evaluate the efficiency in Section~\ref{sec:Efficiency};
and finally, we compare our algorithm with some existing ones in Section~\ref{sec:Comparison}.

For all experiments, we used decay factor $c = 0.6$, as suggested by Lizorkin et al.~\cite{lizorkin2010accuracy},
and the number of SimRank iterations $T = 11$, which is the same as Fogaras and R\'acz~\cite{fogaras2005scaling}.

\begin{table}[tb]
\tbl{Dataset information.\label{tbl:dataset}}{
\centering
\small
\begin{tabular}{l|rr} \hline
  Dataset           &\multicolumn{1}{c}{$|V|$} & \multicolumn{1}{c}{$|E|$} \\ \hline
ca-GrQc           &     5,242 &       14,496  \\
as20000102        &     6,474 &       13,895  \\
Wiki-Vote         &     7,155 &      103,689  \\
ca-HepTh          &     9,877 &       25,998  \\
email-Enron       &    36,692 &      183,831  \\
soc-Epinions1     &    75,879 &      508,837  \\
soc-Slashdot0811  &    77,360 &      905,468  \\
soc-Slashdot0902  &    82,168 &      948,464  \\
email-EuAll       &   265,214 &      400,045  \\
web-Stanford      &   281,903 &    2,312,497  \\
web-NotreDame     &   325,728 &    1,497,134  \\
web-BerkStan      &   685,230 &    7,600,505  \\
web-Google        &   875,713 &    5,105,049  \\
dblp-2011         &   933,258 &    6,707,236  \\
in-2004           & 1,382,908 &   17,917,053  \\
flickr            & 1,715,255 &   22,613,981  \\
soc-LiveJournal   & 4,847,571 &   68,993,773  \\
indochina-2004    & 7,414,866 &  194,109,311  \\
it-2004           &41,291,549 &1,150,725,436  \\
twitter-2010      &41,652,230 &1,468,365,182  \\
uk-2007-05       &105,896,555 &3,738,733,648  \\ \hline
\end{tabular}}
\end{table}

All experiments were conducted on an Intel Xeon E5-2690
2.90GHz CPU with 256GB memory running Ubuntu 12.04.
Our algorithm was implemented in C++ and was compiled
using g++v4.6 with the -O3 option.

\subsubsection{Accuracy}
\label{sec:Accuracy}

The accuracy of our framework depends on the accuracy of the estimated diagonal correction matrix,
computed via Algorithm~\ref{alg:DiagonalEstimation}.
As discussed in Section~\ref{sec:accuracy}, our algorithm has two parameters, $L$ and $R$,
the number of iterations for the Gauss-Seidel method,
and the number of samples for Monte Carlo simulation, respectively.
We evaluate the accuracy by changing these parameters.

To evaluate the accuracy,
we first compute the \emph{exact} SimRank matrix $S$ by Jeh and Widom's original algorithm~\cite{jeh2002simrank},
and then compute the \emph{mean error}~\cite{yu2012space} defined as follows:
\begin{align}
\label{eq:ME}
  \textrm{ME} = \frac{1}{n^2} \sum_{i,j} \left|S^L(D)_{ij} - s(i,j)\right|.
\end{align}
Since this evaluation is expensive (i.e., it requires SimRank scores for $O(n^2)$ pairs),
we used the following smaller datasets:
ca-GrQc,
as20000102,
wiki-Vote,
and
ca-HepTh.
Results are shown in Figure~\ref{fig:accuracy},
and we summarize our results below.
\begin{itemize}
\item
For mean error ME $\le$ $10^{-5} \sim 10^{-6}$,
we need only $R = 100$ samples with $L = 3$ iterations.
Note that this is the same accuracy level as~\cite{yu2012space}.

\item
If we want more accurate SimRank scores,
we need much more samples $R$
and little more iterations $L$.
This coincides with the analysis in Section~\ref{sec:accuracy}
in which we estimated
$L = O(\log (n / \epsilon))$ and $R = O((\log n)/\epsilon^2)$.
\end{itemize}


%
%

\begin{figure}[ht]
  \centering

  \subfigure[ca-GrQc]{\epsfig{file=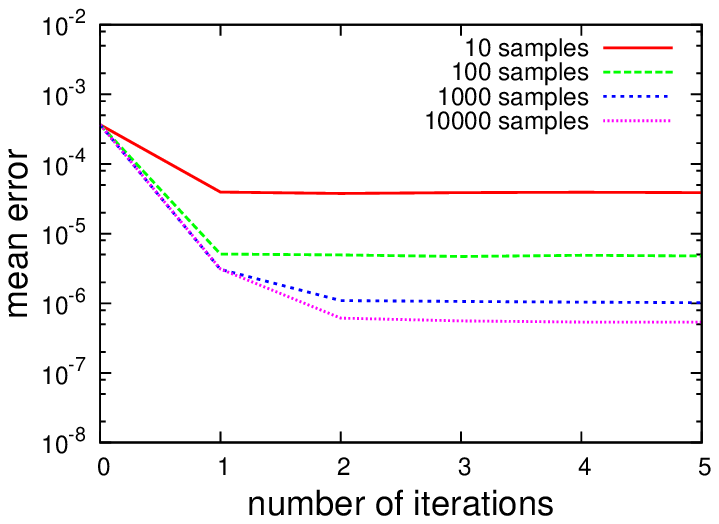, width=13.5em}}
  \subfigure[as20000102]{\epsfig{file=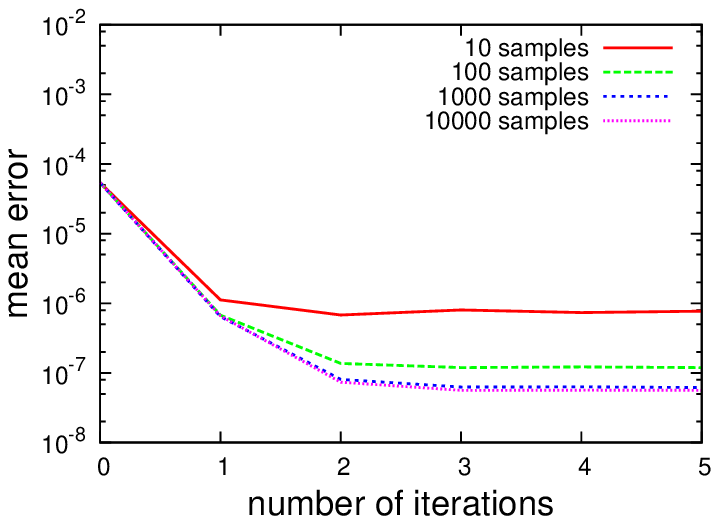, width=13.5em}}

  \subfigure[wiki-Vote]{\epsfig{file=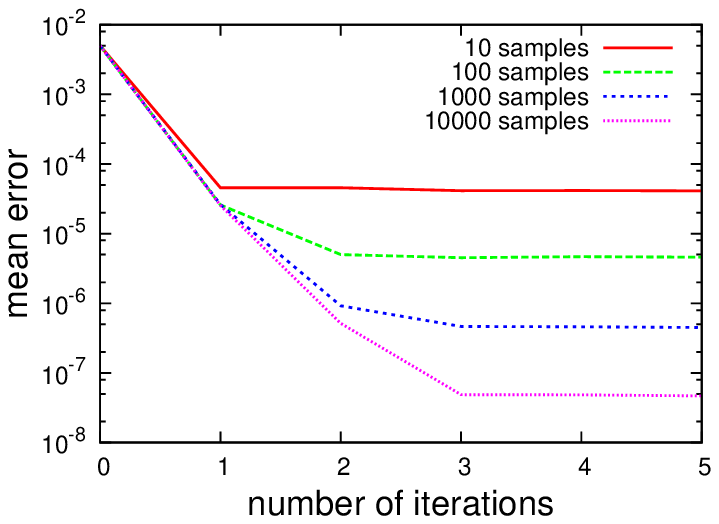, width=13.5em}}
  \subfigure[ca-HepTh]{\epsfig{file=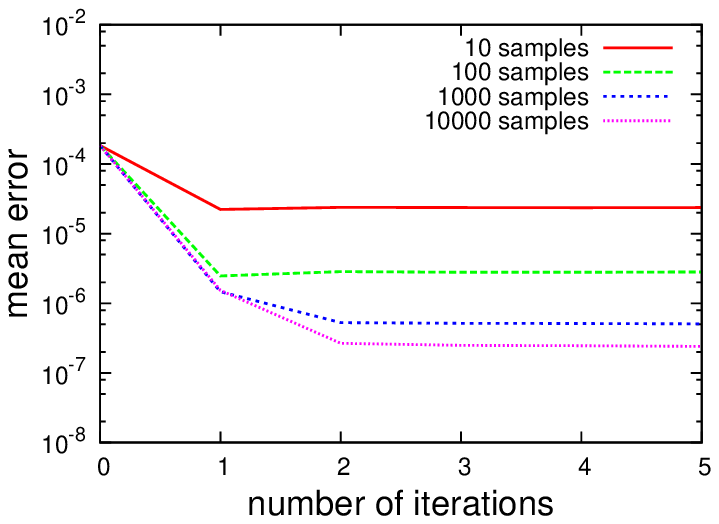, width=13.5em}}
  \caption{Number of iterations $L$ vs. mean error of computed SimRank.} \label{fig:accuracy}
\end{figure}

\subsubsection{Efficiency}
\label{sec:Efficiency}

We next evaluate the efficiency of our algorithm.
We first performed preprocessing with parameters $R = 100$ and $L = 3$, respectively. 
We then performed single-pair, single-source and all-pairs queries for real networks.
Results are shown in Table~\ref{tbl:result};
we omitted results of the all-pairs computation for a network larger than in-2004
since runtimes exceeded three days.
We summarize our results below.

\begin{itemize}
\item
For small networks ($n \le 1,000,000$),
only a few minutes of preprocessing time were required;
furthermore,
answers to single-pair queries were obtained in 100 milliseconds,
while answers to single-source queries were obtained in 300 milliseconds.
This efficiency is certainly acceptable for online services.
We were also able to solve all-pairs query in a few days.

\item
For large networks, $n \ge 40,000,000$,
a few hours of preprocessing time were required;
furthermore,
answers to single-pair queries were obtained approximately in 10 seconds,
while answers to single-source queries were obtained in a half minutes.
To the best of our knowledge, this is the first time that such an algorithm is 
successfully scaled up to such large networks.

\item
The space complexity is proportional to the number of edges,
which enables us to compute SimRank values for large networks.

\end{itemize}

\begin{table}
\tbl{Computational results of our proposed algorithm and existing algorithms; single-pair and single-source results are the average of 10 trials;
we omitted results of the all-pairs computation of our proposed algorithm for a network larger than in-2004 since runtimes exceeded three days;
other omitted results (---) mean that the algorithms failed to allocate memory.  \label{tbl:result}}{
\small
\centering
\tabcolsep=3pt
\begin{tabular}{l|rrrrr|rr|rrrr} \hline
  Dataset           & \multicolumn{5}{c|}{Proposed} & \multicolumn{2}{c|}{\cite{yu2012space}} & \multicolumn{4}{c}{\cite{fogaras2005scaling}} \\
                    & Preproc.& SinglePair & SingleSrc. & AllPairs & Memory & AllPairs & Memory & Preproc. & SinglePair & SingleSrc. & Memory \\ \hline
  ca-GrQc           &   842 ms&  0.236 ms &  2.080 ms &    10.90 s &    3 MB & 2.97 s & 69 MB   & 64.7 s & 87.0 ms &  288 ms & 23.3 GB \\
  as20000102        &    96 ms&  0.518 ms &  0.812 ms &     5.26 s &    2 MB & 0.13 s & 8 MB    &  106 s &  112 ms &  262 ms & 28.1 GB \\
  Wiki-Vote         &   187 ms&  0.613 ms &  5.26 ms &      37.4 s &    6 MB & 8.74 s & 143 MB  & 43.4 s & 30.4 ms & 42.5 ms & 24.5 GB \\
  ca-HepTh          &   698 ms&  0.493 ms &  3.24 ms &      39.0 s &   4 MB & 23.3 s & 316 MB   &  205 s & 61.2 ms &  262 ms & 44.2 GB \\
  email-Enron       &   2.75 s&  2.56 ms &  24.13 ms &      885 s &   20 MB & 302 s & 3.47 GB   & 8055 s & 86.9 ms & 1.19 ms & 162 GB \\
  soc-Epinions1     &   4.12 s&  5.90 ms &  74.4 ms &      5647 s &   31 MB & 777 s & 6.94 GB   & --- & --- & ---  & --- \\
  soc-Slashdot0811  &   6.14 s&  4.16 ms &  20.4 ms &      1581 s &   47 MB & 747 s & 7,37 GB   & --- & --- & --- & --- \\
  soc-Slashdot0902  &   5.87 s&  4.63 ms &  21.0 ms &      1725 s &   49 MB & 694 s & 7.24 GB   & --- & --- & --- & --- \\
  email-EuAll       &   11.3 s&  14.5 ms &  61.7 ms &     4.54 h &   57 MB & 2.00 h & 59.1 GB   & --- & --- & --- & --- \\ 
  web-Stanford      &   21.0 s&  9.76 ms &   288 ms &     22.5 h &  132 MB & --- & --- &          --- & --- & --- & --- \\ 
  web-NotreDame     &   8.07 s&  14.7 ms &  47.3 ms &     4.28 h &  107 MB & 1.50 h & 45.5 GB   & --- & --- & --- & --- \\ 
  web-BerkStan      &   49.6 s&  35.7 ms &   272 ms &     51.7 h &  392 MB & ---    & ---       & --- & --- & --- & --- \\ 
  web-Google        &   52.2 s&  64.2 ms &   234 ms &     57.0 h &  325 MB & 11.1 h & 203 GB    & --- & --- & --- & --- \\ 
  dblp-2011         &    104 s&  53.6 ms &   207 ms &     53.7 h&  395 MB  & 3140 s & 24.1 GB   & --- & --- &  --- & --- \\ 
  in-2004           &   71.7 s&  91.1 ms &   335 ms &    ---   &  843 MB   & --- & ---          & --- & --- & --- & --- \\
  flickr            &    160 s&   137 ms &   424 ms &    ---   & 1.11 GB   & --- & ---          & --- & --- & --- & --- \\
  soc-LiveJournal   &    819 s&   394 ms &  1.19 s &     ---   & 3.74 GB   & --- & --- & --- & --- & --- & --- \\
  indochina-2004    &    391 s&   487 ms &  1.73 s &     ---   & 8.15 GB   & --- & --- & --- & --- & --- & --- \\ 
  it-2004           &   2822 s&   3.51 s&   12.0 s&      ---   & 49.2 GB   & --- & --- & --- & --- & --- & --- \\
  twitter-2010      &  14376 s&   3.17 s&   11.9 s&      ---   & 59.4 GB   & --- & --- & --- & --- & --- & --- \\
  uk-2007-05        &   8291 s&   9.42 s&   32.7 s&      ---   &  153 GB   & --- & --- & --- & --- & --- & --- \\ \hline
\end{tabular}}
\end{table}

\subsubsection{Comparisons with existing algorithms}
\label{sec:Comparison}

In this section,
we compare our algorithm with two state-of-the-art algorithms for computing SimRank.
We used the same parameters ($R = 100$, $L = 3$) as the above.

\textbf{Comparison with the state-of-the-art all-pairs algorithm}

Yu et al.~\cite{yu2012space} proposed an efficient all-pairs algorithm;
the time complexity of their algorithm is $O(T n m)$,
and the space complexity is $O(n^2)$.
They computed SimRank via matrix-based iteration \eqref{eq:simrank}
and reduced the space complexity by discarding entries in SimRank matrix
that are smaller than a given threshold.
We implemented their algorithm and evaluated it in comparison with ours.
We used the same parameters presented in \cite{yu2012space}
that attain the same accuracy level as our algorithm. 

Results are shown in Table~\ref{tbl:result};
the omitted results (---) mean that their algorithm failed to allocate memory.
From the results, we observe that their algorithm performs
a little faster than ours,
because the time complexity of their algorithm is $O(T n m)$,
whereas the time complexity of our algorithm is $O(T^2 n m)$;
however, our algorithm uses much less space.
In fact, their algorithm failed for a network with $n \ge $ 300,000 vertices
because of memory allocation.
More importantly, their algorithm cannot estimate the memory usage
before running the algorithm.
Thus, our algorithm significantly outperforms their algorithm in terms of scalability.


\textbf{Comparison with the state-of-the-art single-pair and single-source algorithm}

Fogaras and R\'acz~\cite{fogaras2005scaling} proposed an efficient single-pair algorithm
that estimates SimRank scores by using first meeting time formula \eqref{eq:randomsurferpair}
with Monte Carlo simulation.
Like our approach, their algorithm also consists of two phases, a preprocessing phase and a query phase.
In the preprocessing phase, their algorithm generates $R'$ random walks and stores the walks efficiently;
this phase requires $O(n R')$ time and $O(n R')$ space.
In the query phase, their algorithm computes scores via formula \eqref{eq:randomsurferpair};
this phase requires $O(T n R')$ time.
We implemented their algorithm and evaluated it in comparison with ours.

We first checked the accuracy of their algorithm by computing all-pairs SimRank for the smaller datasets used in Section~\ref{sec:Accuracy};
results are shown in Table~\ref{tbl:FogarasRaczAccuracy}.
From the table, we observe that in order to obtain the same accuracy as our algorithm,
their algorithm requires $R' \ge$ 100,000 samples,
which are much larger than our random samples $R = 100$.
This is because
their algorithm estimates all $O(n^2)$ entries by Monte Carlo simulation,
but our algorithm only estimates $O(n)$ diagonal entries by Monte Carlo simulation.

We then evaluated the efficiency of their algorithm with $R' = $ 100,000 samples. 
These results are shown in Table~\ref{tbl:result}.
This shows that their algorithm needs much more memory,
thus it only works for small networks.
This concludes that in order to obtain accurate scores,
our algorithm is much more efficient than their algorithm.

%
%
%

\begin{table}[tb]
\tbl{Accuracy of the single-pair algorithm proposed by Fogaras and R\'acz~\cite{fogaras2005scaling}; accuracy is shown as mean error. \label{tbl:FogarasRaczAccuracy}}{
\small
\centering
\begin{tabular}{l|rrl} \hline
  Dataset & Samples & Accuracy \\ \hline
  ca-GrQc    & 100    &  1.59 $\times 10^{-4}$\phantom{)} \\
             & 1,000   &  5.87 $\times 10^{-5}$\phantom{)} \\
             & 10,000  &  1.32 $\times 10^{-5}$\phantom{)} \\
             & 100,000 &  6.43 $\times 10^{-6}$\phantom{)} \\
             & (Proposed & 4.77 $\times 10^{-6}$) \\ \hline
  as20000102 & 100    &  2.51 $\times 10^{-3}$\phantom{)} \\
             & 1,000   &  7.87 $\times 10^{-4}$\phantom{)} \\
             & 10,000  &  2.54 $\times 10^{-4}$\phantom{)} \\
             & 100,000 &  8.69 $\times 10^{-5}$\phantom{)} \\
             & (Proposed & 1.19 $\times 10^{-7}$) \\ \hline
  wiki-Vote  & 100    &  1.03 $\times 10^{-3}$\phantom{)} \\
             & 1,000   &  3.57 $\times 10^{-4}$\phantom{)} \\
             & 10,000  &  1.13 $\times 10^{-4}$\phantom{)} \\
             & 100,000 &  3.63 $\times 10^{-5}$\phantom{)} \\
             & (Proposed & 2.81 $\times 10^{-6}$) \\ \hline
  ca-HepTh   & 100    &  1.36 $\times 10^{-4}$\phantom{)} \\
             & 1,000   &  5.58 $\times 10^{-5}$\phantom{)} \\
             & 10,000  &  1.18 $\times 10^{-5}$\phantom{)} \\
             & 100,000 &  6.04 $\times 10^{-6}$\phantom{)} \\
             & (Proposed & 4.56 $\times 10^{-6}$) \\ \hline
\end{tabular}}
\end{table}

\clearpage

\section{Top \textit{k}-computation$^1$}
\footnotetext[1]{This section is based on our SIGMOD'14 paper~\cite{kusumoto2014scalable}}
\label{sec:topk}

\subsection{Motivation and overview}

In the previous section, we describe the SimRank linearization technique
and show how to use the linearization to solve single-pair, single-source, and all-pairs SimRank problems.
In this section, we consider an top $k$ search problem;
we are given a vertex $i$ and then find $k$ vertices with the $k$ highest SimRank scores with respect to $i$.
This problem is interested in many applications because,
usually, highly similar vertices for a given vertex $i$ are very few (e.g., 10--20),
and in many applications, we are are only interested in such highly similar vertices.
This problem can be solved in $O(T^2 m)$ time by using the single-source SimRank algorithm;
however, since we only need $k$ highly-similar vertices, 
we can develop more efficient algorithm.

Here, we propose a Monte-Carlo algorithm based on SimRank linearization for this problem;
the complexity is independent of the size of networks.
Note that Fogaras and Racz~\cite{fogaras2005scaling} also proposed the Monte-Carlo based single-pair computation algorithm.
By comparing their algorithm, our main ingredients is the ``pruning technique'' by utilizing the SimRank linearization.
We observe that SimRank score $s(i,j)$ decays very rapidly as distance of the pair $i, j$ increases.
To exploit this phenomenon, we establish upper bounds of SimRank score $s(i,j)$ that only depend on distance $d(i,j)$. The upper bounds can be efficiently computed by Monte-Carlo simulation (in our preprocess).
These upper bounds, together with some adaptive sample technique, allow us to effectively prune the similarity search procedure.


Overall, the proposed algorithm runs as follows.
\begin{enumerate}
\item 
We first perform preprocess to compute the auxiliary values
for upper bounds of SimRank $s(i,j)$ for all $i \in V$ (see Section~\ref{sec:upperbound}).
In addition, we construct an auxiliary bipartite graph $H$, which allows us to enumerate 
``candidates'' of highly similar vertices $j$ more accurate.
This is our preprocess phase. The time complexity is $O(n)$.

\item 
We now perform our query phase.
We compute SimRank scores $s(i,j)$ by the Monte-Carlo simulation 
for each vertex $i$ in the ascending order of distance from a given vertex $i$, 
and at the same time, we perform ``pruning'' by the upper bounds.
We also combine the adaptive sampling technique for the Monte-Carlo simulation;
specifically, for a query vertex $i$, we first estimate SimRank scores roughly for each candidate $j$
by using a small number of Monte-Carlo samples,
and then we re-compute more accurate SimRank scores for each candidate $j$ that has high estimated SimRank scores.
\end{enumerate}

\subsection{Monte Carlo algorithm for single-pair SimRank}
\label{sec:montecarlo}




Let us consider a random walk that starts from $i \in V$ and that follows its in-links,
and let $i^{(t)}$ be a random variable for the $t$-th position of this random walk.
Then we observe that
\begin{align} \label{eq:expectationP}
  P^t e_i = \E[ e_{ i^{(t)}  }].
\end{align}
Therefore, by plugging \eqref{eq:expectationP} to \eqref{eq:ForSinglePair},
we obtain
\begin{align} \label{eq:expectation}
  s^{(T)}(i,j) = & D_{i j} + c \E[ e_{i^{(1)}} ]^\top D \E[ e_{j^{(1)}} ]  \nonumber \\
                   &+ \cdots + c^{T-1} \E[ e_{i^{(T-1)}} ]^\top D \E[ e_{j^{(T-1)}} ].
\end{align}
Our algorithm computes the expectations in the right hand side of \eqref{eq:expectation}
by Monte-Carlo simulation as follows:
Consider $R$ independent random walks $i^{(t)}_1, \ldots, i^{(t)}_R$ that start from $i \in V$,
and $R$ independent random walks $j^{(t)}_1, \ldots, j^{(t)}_R$ that start from $j \in V$ with $i\neq j$.
Then each $t$-th term of \eqref{eq:expectation} can be estimated as
\begin{align} \label{eq:tthestimate}
  c^t \E[ e_{i^{(t)}} ]^\top D \E[ e_{j^{(t)}} ] \simeq \frac{c^t}{R^2} \sum_{r = 1}^R \sum_{r' = 1}^R D_{i_r^{(t)} j_{r'}^{(t)} }.
\end{align}
We compute the right hand side of \eqref{eq:tthestimate}.
Specifically by maintaining the positions of $i_1^{(t)}, \ldots, i_R^{(t)}$ and $j_1^{(t)}, \ldots, j_R^{(t)}$ by hash tables,
it can be evaluated in $O(R)$ time.
Therefore the total time complexity to evaluate \eqref{eq:expectation} is $O(T R)$.
The algorithm is shown in Algorithm~\ref{alg:singlepair}.
We emphasize that this time complexity is independent of the size of networks (i.e, $n, m$)
Hence this algorithm can scale to very large networks.

\begin{algorithm}[tb]
\caption{Monte-Carlo Single-pair SimRank} \label{alg:singlepair}
\begin{algorithmic}[1]
\Procedure{SinglePair}{$i$,$j$}
\State{$i_1 \leftarrow i, \ldots, i_R \leftarrow i$, $j_1 \leftarrow j, \ldots, j_R \leftarrow j$}
\State{$\sigma = 0$}
\For{$t = 0, 1, \ldots, T-1$}
\For{$w \in \{ i_1, \ldots, i_R \} \cap \{ i_1 \ldots, i_R \}$}
\State{$\alpha \leftarrow \#\{ r : i_r = w, r = 1, \ldots, R \}$}
\State{$\beta \leftarrow \#\{ r : j_r = w, r = 1, \ldots, R \}$}
\State{$\sigma \leftarrow \sigma + c^t D_{ww} \alpha \beta / R^2 $}
\EndFor
\For{$r = 1, \ldots, R$}
\State{$i_r \leftarrow \delta_-(i_r)$, $j_r \leftarrow \delta_-(j_r)$, randomly}
\EndFor
\EndFor
\State{\textbf{return} $\sigma$}
\EndProcedure
\end{algorithmic}
\end{algorithm}

We give estimation of the number of samples
to compute \eqref{eq:expectation} accurately, with high probability.
We use the Hoeffding inequality to show the following.

\begin{proposition} \label{prop:bound}
Let $\tilde s^{(T)}(i,j)$ be the output of the algorithm. Then
\begin{align} \label{eq:bound}
  \P \left\{ | \tilde s^{(T)}(i,j) - s^{(T)}(i,j) | \ge \epsilon \right\} \le 4 n T \exp \left( - \epsilon^2 R / 2(1-c)^2 \right).
\end{align}
\end{proposition}
\begin{lemma} \label{lem:Duvbound}
\begin{align*}
  \P\left\{ \left| X_u^{(t) \top} D X_v^{(t)} - (P^t e_u)^\top D P^t e_v \right| \ge \epsilon \right\}
  \le 4 n \exp(- \epsilon^2 R / 2).
\end{align*}
\end{lemma}
\begin{proof}
\begin{align*}
& \P\left\{ \left| X_u^{(t) \top} D X_v^{(t)} - (P^t e_u)^\top D P^t e_v \right| \ge \epsilon \right\} \\
&\le \P \left\{ \left| X_u^{(t) \top} D \left( X_v^{(t)} - P^t e_v \right) \right| \ge \epsilon/2 \right\} \\
&+ \P \left\{ \left| \left( X_u^{(t)} - P^t e_u\right)^\top D P^t e_v \right| \ge \epsilon/2 \right\} \\
&\le 4 n \exp(-\epsilon^2 R / 2). 
\end{align*}
\end{proof}
\begin{proof}[of Proposition~\ref{prop:bound}]
By Lemma~\ref{lem:Duvbound}, we have
\begin{align*}
  & \P \left\{ \left| \left( \sum_{t=0}^{T-1} c^t X_u^{(t) \top} D X_v^{(t)} \right) - s^{(T)}(u,v) \right| \ge \epsilon \right\} \\
  \le & \sum_{t = 0}^{T-1} \P \left\{ \left| c^t X_u^{(t) \top} D X_v^{(t)} - c^t (P^t e_u)^\top D P^t e_v \right| \ge c^t \epsilon/(1-c)\right\} \\
  \le & 4 n T \exp \left( - \epsilon^2 R / 2(1-c)^2\right). 
\end{align*}
\end{proof}

By Proposition~\ref{prop:bound},
we have the following.
\begin{corollary}
Algorithm~\ref{alg:singlepair} computes the SimRank score $s^{(T)}(u,v)$ with accuracy $0 < \epsilon < 1$
with probability $0 < \delta < 1$ by setting $R = 2(1-c)^2 \log(4 n T / \delta) / \epsilon^2$.
\end{corollary}

%

\subsection{Distance correlation of SimRank}
\label{sec:correctness}

Our top $k$ similarity search algorithm performs
single-pair SimRank computations
for a given source vertex $i$ and for other vertices $j$, but we save the time complexity
by pruning. In order to perform this pruning, we need some upper bounds.
This section and the next section are devoted to establish the upper bounds.

The important observation of SimRank is
\begin{quote}
SimRank score $s(i,j)$ decays very fast as the pair $i, j$ goes away.
\end{quote}
In this section, we empirically verify this fact in some real networks,
and in the next section, we develop the upper bounds that only depend on distance.

Let us look at Figure~\ref{fig1}. We randomly chose 100 vertices $i$ and
enumerate top-1000 similar vertices with respect to to a query vertex $u$ (note that these top-1000 vertices are
``exact'', not `approximate'').
Each point denotes the average distance of the $k$-th similar vertex.
To convince the reader,
we also give the average distance between two vertices
for each network by the blue line.

Figure~\ref{fig1} clearly shows much intuitive information.
If we only need to compute top-10 vertices, all of them are within distance two, three, or four.
In real applications, it is unlikely that
we need to compute top-1000 vertices, but even for this case, most of them are within distance four or five.
We emphasize that these distances are smaller than the average distance of two vertices in each network.
Thus we can conclude that
the ``candidates'' of highly similar vertices
are screened by distances very well.

%

\begin{figure}[ht]
  \centering
  \begin{minipage}[b]{.49\linewidth}
    \includegraphics[width=4cm,bb=50 50 410 302]{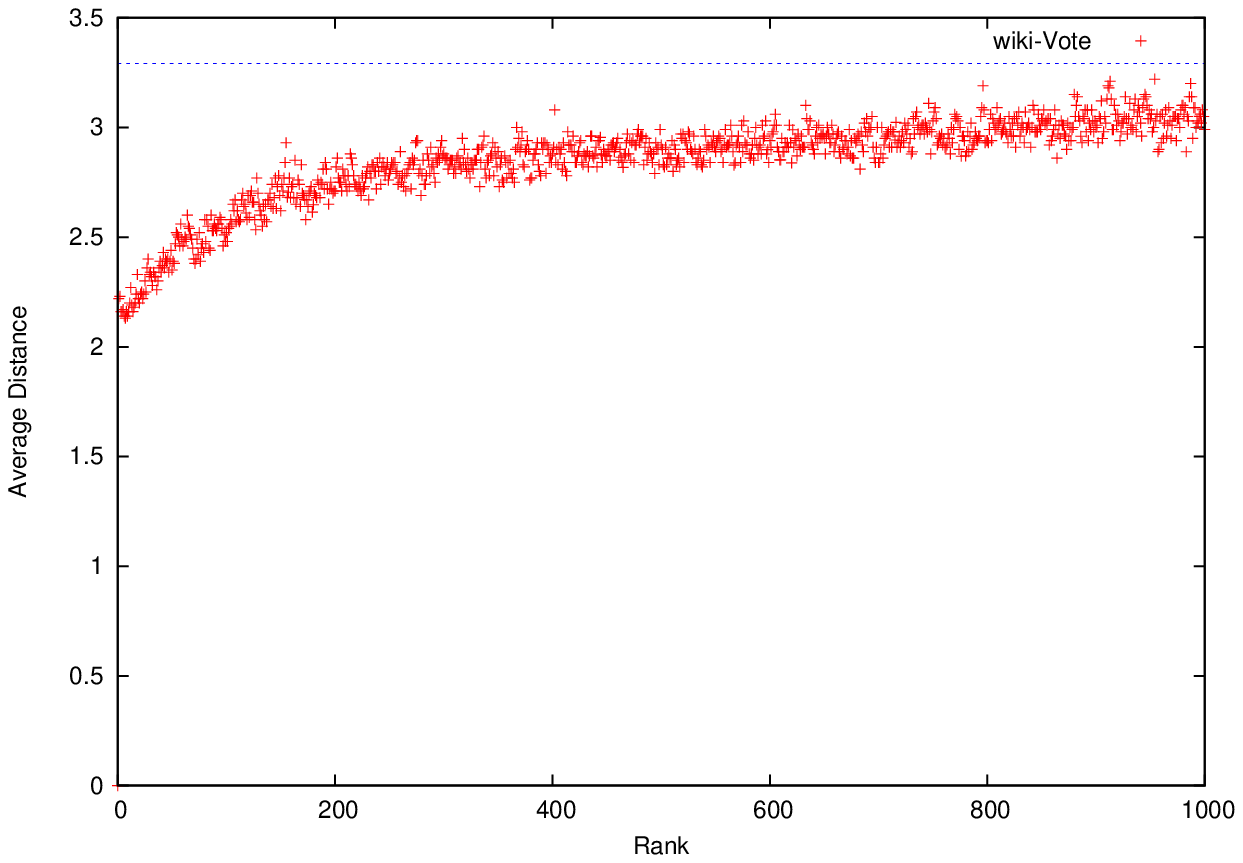}
    \subcaption{wiki-Vote}
  \end{minipage}
  \begin{minipage}[b]{.49\linewidth}
    \includegraphics[width=4cm,bb=50 50 410 302]{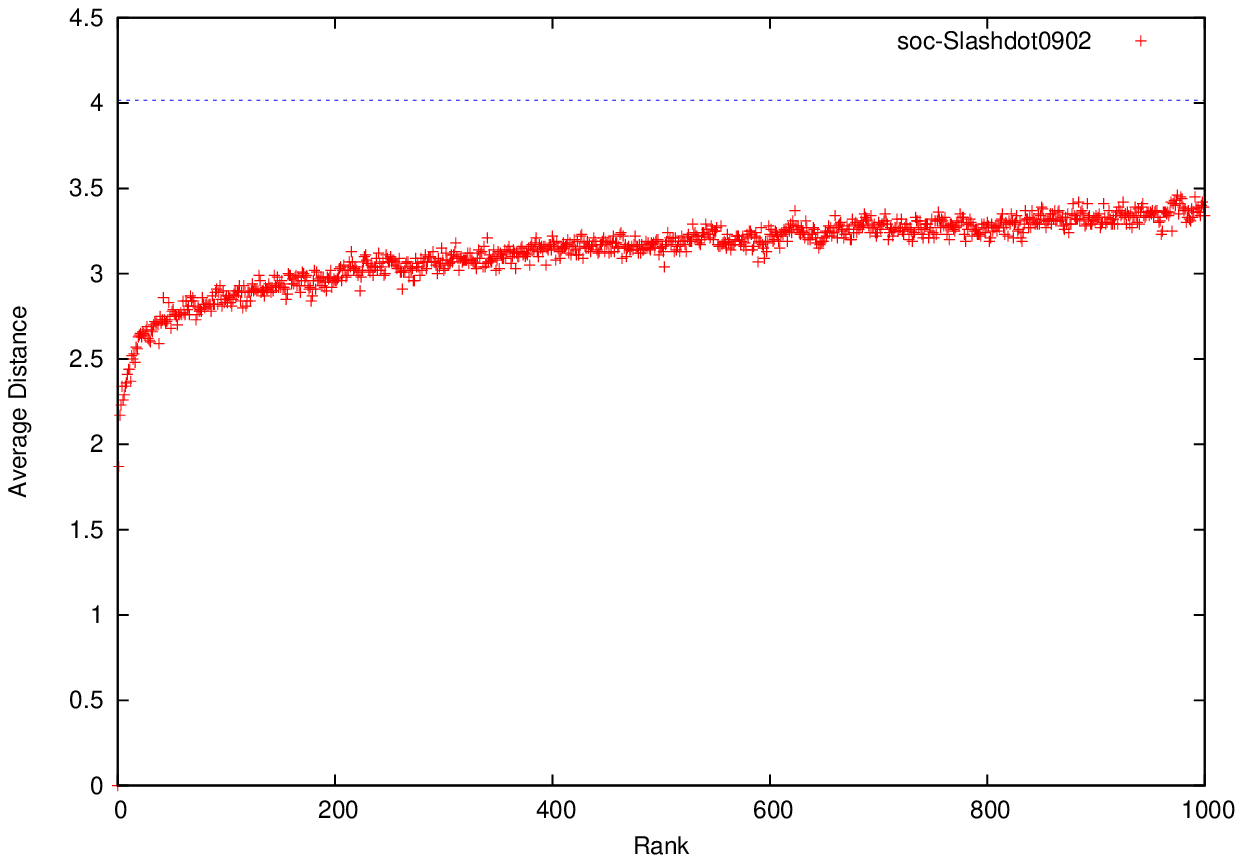}
    \subcaption{soc-Slashdot0902}
  \end{minipage}

  \begin{minipage}[b]{.49\linewidth}
    \includegraphics[width=4cm,bb=50 50 410 302]{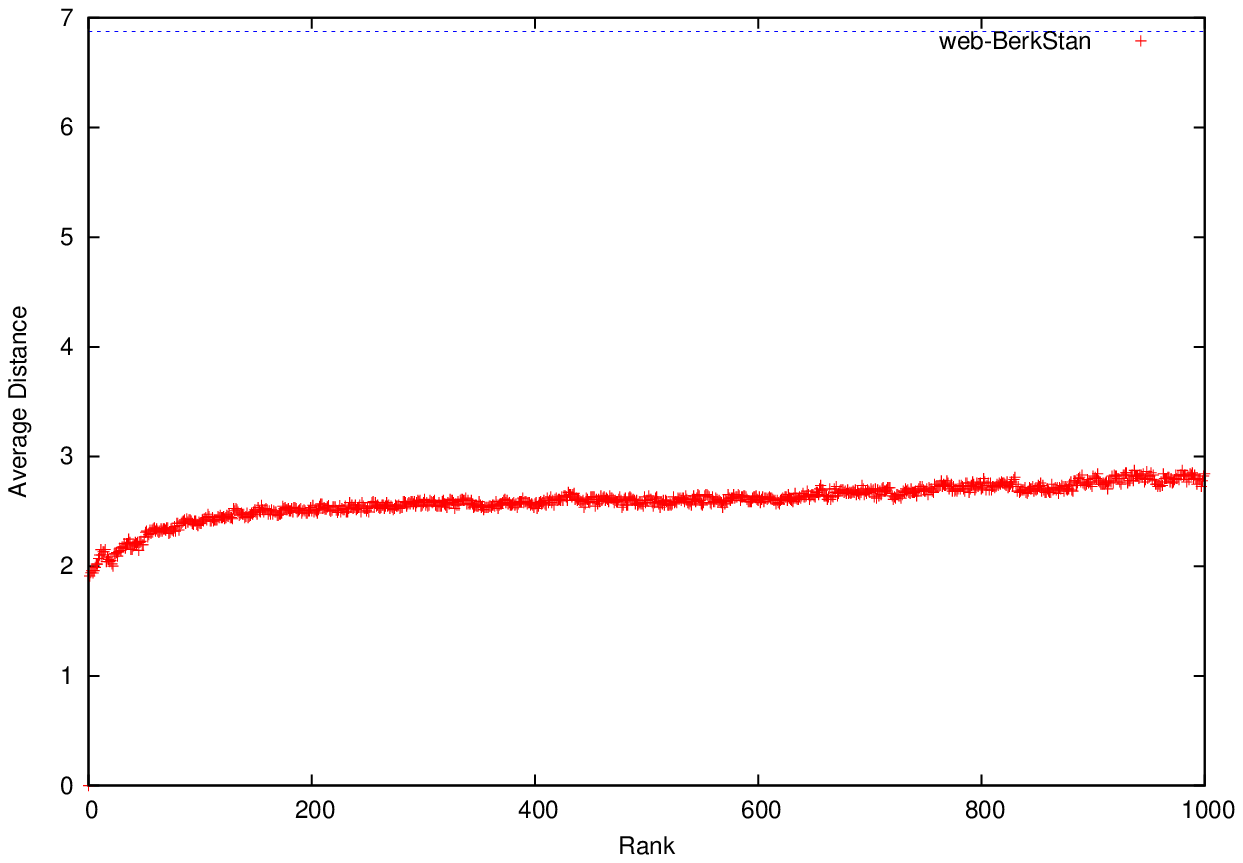}
    \subcaption{web-BerkStan}
  \end{minipage}
  \begin{minipage}[b]{.49\linewidth}
    \includegraphics[width=4cm,bb=50 50 410 302]{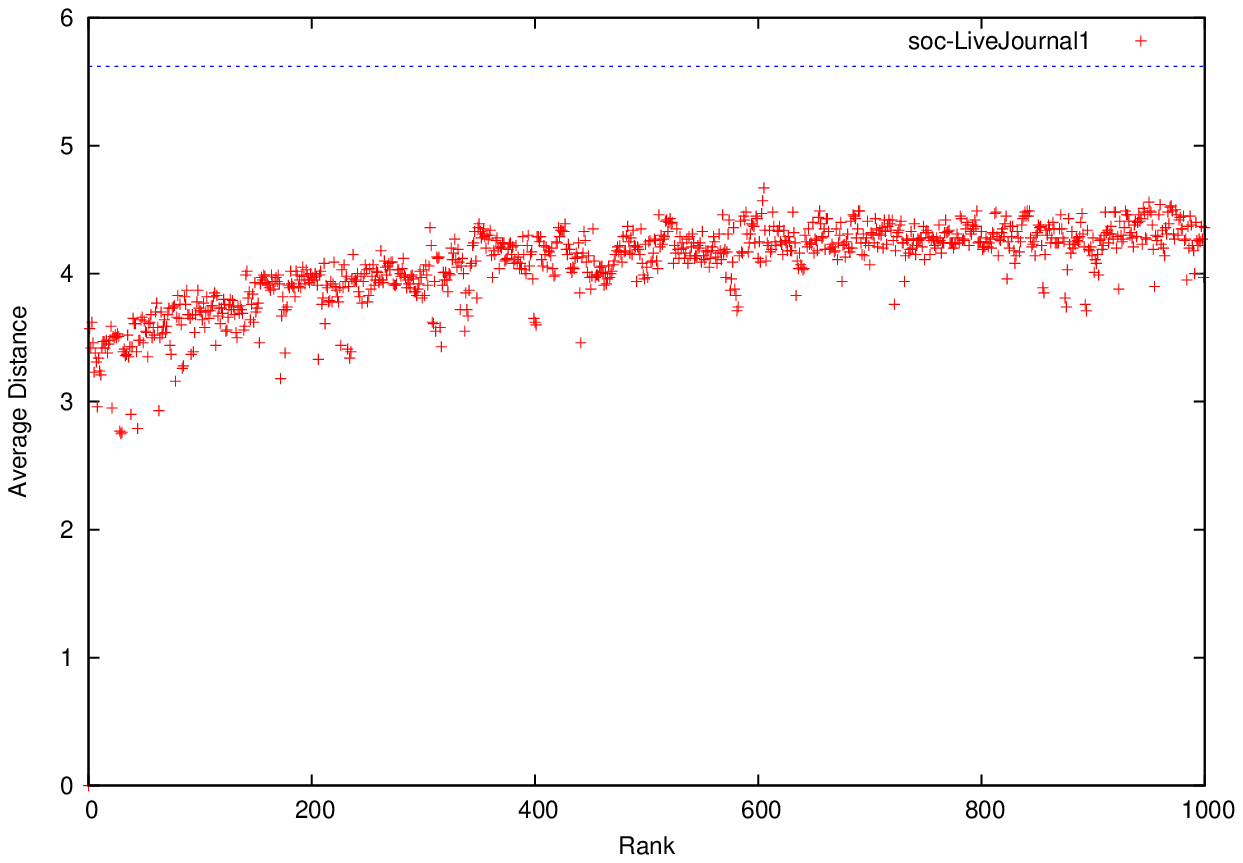}
    \subcaption{soc-LiveJournal1}
  \end{minipage}
  \caption{Distance correlation of similarity ranking. The red points are for distance of top-1000 similar vertices and the blue line is for average distance of two vertices in each network}
  \label{fig1}
\end{figure}

There is one remark we would like to make from Figure~\ref{fig1}. The top-10 highest SimRank vertices in Web graphs
are much closer to a query vertex than social networks. Thus we can also claim that our algorithm would work better for Web graphs than for social networks, because we only look at subgraphs induced by vertices of distance within three (or even two) from a query vertex.
This claim is verified in Section~\ref{sec:experiment}.

\subsection{Tight upper bounds}
\label{sec:upperbound}

In the previous section, we observe that
highly similar vertices with respect to a query vertex are within small distance from $u$.
This observation allows us to propose our efficient algorithm
for the top-$k$ similarity search problem for a single vertex. 
In order to obtain this algorithm,
%
we need to establish the upper bounds of SimRank
that depend only on distance, which will be done in this section.

%
%
%

Let us observe that by definition, SimRank score is bounded by
the decay factor to the power of the distance:
\[
  s(u,v) \le c^{d(u,v)}
\]
Since almost all high SimRank score vertices with respect to a query vertex $u$
are located within distance three from $u$ (see Figure~\ref{fig1}),
we obtain $s(u,v) \le c^3 = 0.216$. But this is too large for our purpose (indeed,
our further experiments to compare actual SimRank scores with this bound confirm that it is too large).

Here we propose two upper bounds, called ``L1 bound'' and ``L2 bound''.
Our algorithm, described in a later section, combines these two bounds to perform ``pruning'', which results in
a much faster algorithm.

\subsubsection{L1 bound}

The first bound is based on the following inequality:
for a vector $x$ and a stochastic vector $y$,
\begin{align} \label{eq:l1}
  x^\top y \le \max_{w \in \mathrm{supp}(y)} x^\top e_w,
\end{align}
where $\mathrm{supp}(y) := \{ w \in V : y(w) > 0 \}$ is a positive support of $y$.
We bound $(P^t e_u)^\top D (P^t e_v)$ by this inequality.

Fix a query vertex $u$.
Let us define
\begin{align} \label{eq:alpha}
  \alpha(u, d, t) := \max_{w \in V, d(u,w) = d} (P^t e_u)^\top D e_w
\end{align}
for $d = 1, \ldots, d_{\mathrm{max}}$ and $t = 1, \ldots, T$,
and
\begin{align} \label{eq:beta}
  \beta(u, d) := \sum_{t=0}^{T-1} c^t \max_{d - t \le d' \le d + t} \alpha(u, d', t)
\end{align}
for $t = 1, \ldots, T$. Here $d_\mathrm{max}$ is distance such that
if $d(u,v) > d_\mathrm{max}$ then $s(u,v)$ is too small to take into account.
(We usually set $d_\mathrm{max} = T$).
\begin{proposition} \label{prop:L1bound}
For a vertex $v$ with $d(u, v) = d$, we have
\begin{align} \label{eq:upperbound1}
  s^{(T)}(u,v) \le \beta(u, d)
\end{align}
\end{proposition}
The proof will be given in Appendix.

\begin{remark}
$\alpha(u, d, t)$ has the following probabilistic representation:
\[
  \alpha(u, d, t) = \max_{d(u,w) = d} D_{ww} \P \{ u^{(t)} = w \}
\]
where $u^{(t)}$ denotes the position of a random walk
that starts from $u$ and follows its in-links.
\end{remark}

To compute $\alpha(u,d,t)$ and $\beta(u,d)$,
we can use Monte-Carlo simulation for $P^t e_u$
as shown in Algorithm~\ref{alg:alpha}.

Similar to Proposition~\ref{prop:bound},
we obtain the following proposition, whose proof will be given in Appendix.
This proposition shows that Algorithm~\ref{alg:alpha} can compute $\alpha(u,d,t)$ and $\beta(u,d)$.
\begin{proposition} \label{prop:L1sample}
Let $\tilde \beta(u,d)$ be computed by Algorithm~\ref{alg:alpha}. Then
\begin{align*} 
  \P \left\{ | \tilde \beta(d,t) - \beta(d,t) | \ge \epsilon \right\} \le 2 n d_{\mathrm{max}} T \exp(- 2 \epsilon^2 R)
\end{align*}
\end{proposition}
By Proposition~\ref{prop:L1sample}, we have the following.
\begin{corollary}
Algorithm~\ref{alg:alpha} computes $\beta(u,d)$ with accuracy less than $0 < \epsilon < 1$
with probability at least $0 < \delta < 1$ by setting $R = \log(2 n d_\mathrm{max} T/\delta) / (2 \epsilon^2)$.
\end{corollary}

\begin{algorithm}[ht]
\caption{Monte-Carlo $\alpha(u,d,t)$, $\beta(u,d)$ computation} \label{alg:alpha}
\begin{algorithmic}[1]
\Procedure{ComputeAlphaBeta}{$u$}
\State{$u_1 \leftarrow u, \ldots, u_R \leftarrow u$}
\For{$t = 0, 1, \ldots, T-1$}
  \For{$w \in \{ u_1, \ldots, u_R \}$}
  \State{$\mu \leftarrow D_{ww} \#\{ r : u_r = w, r \in [1, R]\} / R $}
    \State{$\alpha(u, d(u,w), t) \leftarrow \max \{ \alpha( u, d(u,w), t ), \mu \}$}
  \EndFor
  \For{$r = 1, \ldots, R$}
    \State{$u_r \leftarrow \delta(u_r)$ randomly}
  \EndFor
\EndFor
\For{$d = 1, \ldots, T$}
  \State{$\beta(u,d) = \sum_{t=0}^{T-1} c^t \max_{d - t \le d' \le d + t} \alpha(u, d', t)$}
\EndFor
\EndProcedure
\end{algorithmic}
\end{algorithm}

\subsubsection{L2 bound}

The second bound is based on the Cauchy--Schwartz inequality:
for nonnegative vectors $x$ and $y$,
\begin{align} \label{eq:l2}
  x^\top y \le \|x\| \|y\|.
\end{align}
We also bound $(P^t e_u) D (P^t e_v)$ by this inequality.

Let
\begin{align} \label{eq:gamma}
  \gamma(u,t) := \| \sqrt{D} P^t e_u \|,
\end{align}
where $\sqrt{D} = \mathrm{diag}(\sqrt{D_{11}}, \ldots, \sqrt{D_{nn}})$.
Note that, since $D$ is a nonnegative diagonal matrix, $\sqrt{D}$ is well-defined.

\begin{proposition} \label{prop:L2bound}
For two vertices $u$ and $v$, we have
\begin{align} \label{eq:upperbound2}
  s^{(T)}(u,v) \le \sum_{t=0}^T c^t \gamma(u,t) \gamma(v,t)
\end{align}
\end{proposition}
The proof of Proposition~\ref{prop:L2bound} will be given in Appendix.

To compute $\gamma(u, d)$ for each $u$,
we can use Monte-Carlo simulation.
Let us emphasize that we can compute $\gamma(u, d)$ for each $u$ and $d \le d_{\mathrm{max}}$ in
preprocess.

The following proposition, whose proof will be given in Appendix, shows that Algorithm~\ref{alg:gamma}
can compute $\gamma(u, d)$.
\begin{proposition} \label{prop:L2sample}
Let $\tilde \gamma(u,t)$ be computed by Algorithm~\ref{alg:gamma}. Then
\begin{align*} 
  \P \left\{ | \tilde \gamma(u,t) - \gamma(u,t) | \ge \epsilon \right\} \le 4 n \exp \left( -\epsilon^2 R / 8 \right).
\end{align*}
\end{proposition}
The proof of of Proposition~\ref{prop:L2sample} will be given in Appendix. By
Proposition~\ref{prop:L2sample}, we have the following.
\begin{corollary}
Algorithm~\ref{alg:gamma} computes $\gamma(u,t)$ with accuracy less than $0 < \epsilon < 1$
with probability at least $0 < \delta < 1$ by setting $R = 8 \log(4 n/\delta) / \epsilon^2$.
\end{corollary}

\begin{algorithm}[ht]
\caption{Monte-Carlo $\gamma(u,t)$ computation} \label{alg:gamma}
\begin{algorithmic}[1]
\Procedure{ComputeGamma}{$u$}
\State{$u_1 \leftarrow u, \ldots, u_R \leftarrow u$}
\For{$t = 0, 1, \ldots, T-1$}
  \State{$\mu = 0$}
  \For{$w \in \{ u_1, \ldots, u_R \}$}
  \State{$\mu \leftarrow \mu + D_{ww} \#\{ r : u_r = w, r \in [1,R] \}^2 / R^2 $}
  \EndFor
  \State{$\gamma(u, t) \leftarrow \sqrt{ \mu }$}
  \For{$r = 1, \ldots, R$}
    \State{$u_r \leftarrow \delta(u_r)$ randomly}
  \EndFor
\EndFor
\EndProcedure
\end{algorithmic}
\end{algorithm}

\subsubsection{Comparison of two bounds}

The reason why we need both L1 and L2 bounds is the following:
The L1 bound is more effective for a low degree query vertex $u$.
This is because if $u$ has low degree then $P^t e_u$ is sparse.
Therefore the bound \eqref{eq:tthbound} becomes tighter.

On the other hand, the L2 bound is more effective for high degree vertex $u$.
This is because if $u$ has high degree then $P^t e_u$ spreads widely,
and hence each entry is small.
Therefore $\| \sqrt{D} P^t e_u \|$ decrease rapidly.
%
%

\subsection{Algorithm}
\label{sec:algorithm}

We are now ready to provide our whole algorithm for top-$k$ similarity search.
Our algorithm consists of two phases: preprocess phase and query phase.
In the query phase, for a given vertex $u$,
we compute single-pair SimRanks $s(u,v)$ for some vertices $v$ that may have high SimRank value (we call such vertices \emph{candidates} that are computed in the preprocess phase), 
and output $k$ highly similar vertices.

In order to obtain similar vertices accurately,
we have to perform many Monte Carlo simulations
in single-pair SimRank computation (Algorithm~\ref{alg:singlepair}).
Thus, the key of our algorithm is the way to reduce the number of candidates that 
are computed in the preprocess phase in Section 7.1.

\subsubsection{Preprocess phase}

In the preprocess phase, we precompute $\gamma$ in \eqref{eq:gamma} for the L2 bound as described in Algorithm~\ref{alg:gamma}.
Note that we compute $\alpha,\beta$ in \eqref{eq:alpha}, \eqref{eq:beta} for the L1 bound in query phase.

After that, for each vertex $u$, we enumerate ``candidates'' of highly similar vertices $v$.
For this purpose, we consider the following auxiliary bipartite graph $H$.
The left and right vertices of $H$ are copy of $V$ (i.e., $H$ has $2 n$ vertices).
Let $u_\mathrm{left}$ be the copy of $u \in V$ in the left vertices and
let $v_\mathrm{right}$ be the copy of $v \in V$ in the right vertices.
There is an edge $(u_\mathrm{left}, v_\mathrm{right})$ if
a random walk that starts from $u$ frequently reaches $v$.
By this construction,
a pair of vertices $u$ and $v$ has high SimRank score if
$u_\mathrm{left}$ and $v_\mathrm{left}$ share many neighbors.
We construct this bipartite graph $H$ by performing Monte-Carlo simulations in the original graph $G$ as follows.
For each vertex $u$, we iterate the following procedure $P$ times
to construct an index for $u$.
We perform a random walk $W_0$ of length $T$ from $u$ in $G$.
We further perform $Q$ random walks $W_1, \ldots, W_Q$ from $u$.
Let $v$ be $t$-th vertex on $W_0$.
Then we put an edge $(u_\mathrm{left}, v_\mathrm{right})$ in $H$ if
there are at least two random walks in $W_1, \dots W_Q$ that contain $v$
at $t$-th step. The whole procedure is described in Algorithm 3 below.
Here, for a random walk $W_j$ and $t \ge 0$, we denote the vertex at the $t$-th step of $W_j$ by $W_{jt}$.

%

In our experiment, we set $P=10$, $T=11$ and $Q=5$.
The time complexity of this preprocess phase is $O(n (R + P Q) T)$,
where $R$ comes from Algorithm 3 and we set $R=100$ in our experiment.
The space complexity is $O(n P)$, but in practice,
since the number of candidates are usually small,
the space is less smaller than this bound.

\begin{algorithm}[tb]
\caption{Proposed algorithm (preprocess)}
\begin{algorithmic}[1]
  \Procedure{Indexing}{}
  \For{$u \in V$}
  \For{$i = 1, \ldots, P$}
    \State{Perform a random walk $W_0$ from $u$}
    \State{Perform random walks $W_1$,$\ldots$,$W_Q$ from $u$}
    \For{$t = 1, \ldots, T$}
      \If{$W_{jt} = W_{kt}$ for some $j \neq k$}
      \State{Add $W_{0t}$ for index of $u$}
      \EndIf
    \EndFor
  \EndFor
  \EndFor
  \EndProcedure
\end{algorithmic}
\end{algorithm}

\subsubsection{Query phase}


We now describe our query phase.
For a given vertex $u$, we first traverse the auxiliary bipartite graphs $H$
and enumerate all the vertices $v$ that share the neighbor in $H$.
We then prune some vertices $v$ by L1 and L2 bounds.
After that, for each candidate $v$,
we compute SimRank scores $s(u,v)$ by Algorithm~\ref{alg:singlepair}.
Finally we output top $k$ similar vertices as the solution of similarity search.

To accelerate the above procedure, we use the adaptive sample technique.
For a query vertex $u$,
we first set $R = 10$ (in Algorithm~\ref{alg:singlepair}) and 
estimate SimRank scores roughly for each candidate $v$
by Monte Carlo simulation
(i.e, we only perform 10 random walks for $v$ by Algorithm~\ref{alg:singlepair}).
Then, we change $R = 100$ and re-compute more accurate SimRank scores
for each candidate $v$ that has high estimated SimRank scores by Monte Carlo simulation (i.e, we perform
100 random walks for $v$ by Algorithm~\ref{alg:singlepair}).
The whole procedure is described in Algorithm 5.

The overall time complexity of the query phase is
$O(R T |S|)$ where $|S|$ is the number of candidates, since
computing  SimRank scores $s(u,v)$ by Algorithm~\ref{alg:singlepair} for
two vertices $u,v$ takes $O(RT)$.
The space complexity is $O(m + nP)$.

\begin{algorithm}[tb]
\caption{Proposed algorithm (query)}
\begin{algorithmic}[1]
\Procedure{Query}{$u$}
\State{Enumerate $S := \{ v | \delta_H(u_\mathrm{left}) \cap \delta_H(v_\mathrm{left}) \neq \emptyset \}$}
\State{Prune $S$ by L1 and L2 bound}
\For{$v \in S$}
 \State{Perform Algorithm~\ref{alg:singlepair} $R=10$ times to roughly estimate $s(u,v)$.}
   \If{The estimated score $s(u,v)$ is not small}
     \State{Perform Algorithm~\ref{alg:singlepair} $R=100$ times to  estimate $s(u,v)$ more accurately}
 \EndIf
\EndFor
\State{Output top $k$ similar vertices}
\EndProcedure
\end{algorithmic}
\end{algorithm}

\subsection{Experiment}
\label{sec:experiment}

We perform our proposed algorithm for several real networks
and evaluate performance of our algorithm. We also compare our algorithm
with some state-of-the-art algorithms.

All experiments are conducted on an Intel Xeon E5-2690
2.90GHz CPU with 256GB memory and running Ubuntu 12.04.
Our algorithm is implemented in C++ and compiled
with g++v4.6 with -O3 option.

According to discussion in the previous section,
we set the parameters as follows:
decay factor $c = 0.6$, $T = 11$,
$R = 100$ for $\gamma$ (Algorithm~\ref{alg:gamma}) and for $s(\cdot,\cdot)$ (Algorithm~\ref{alg:singlepair}),
and $R = 10000$ for $\alpha$ and $\beta$ (Algorithm~\ref{alg:alpha})
that is optimized by pre-experiment\footnote{These values are much smaller than our theoretical estimations. The reason is that Hoeffding bound is not tight in this case.}. We also set $P=10$, $T=11$ and $Q=5$ in
our preprocess phase as in Section 7.1.

In addition, we set $k=20$ since we are only interested in small number
of similar vertices.
To avoid searching vertices of very small SimRank scores,
we set a threshold $\theta = 0.01$ to terminate the procedure when upper bounds become smaller than $\theta$.


We use the datasets shown in Table~\ref{tbl:dataset}.

\subsection{Results}

We evaluate our proposed algorithm for several real networks.
The results are summarized in Tables~\ref{tbl:result}.


We first observe that
our proposed algorithm can find top-20 similar vertices
in less than a few seconds for graphs of billions edges (i.e., ``it-2004'')
and in less than a second for graphs of one hundred millions edges, respectively.


We can also observe that the query time for our algorithm does not much depend on the size of networks.
For example, ``indochina-2004'' has 8 times more edges
than ``flickr'' but
the query time is twice faster than that.
Hence the computational time of our algorithm
depends on the \emph{network structure} rather than the network size.
Specifically, our algorithm works better for web graphs than for social networks.

\subsubsection{Analysis of Accuracy}

In this subsection, we shall investigate performance of
our algorithm in terms of accuracy. In many applications, 
we are only interested in very similar vertices. Hence 
we only look at vertices that have high SimRank scores.

Specifically,
what we do is the following. We first compute, for a query vertex $u$,
the single source SimRank scores $s(u,v)$ for all the vertices $v$ (for the whole graph)
by the exact method. Then we pick up all ``high score'' vertices $v$ with score at least $t$ from this computation (for $t=0.04$, $0.05$, $0.06$, $0.07$).
Finally, we compute ``high score'' vertices $v$
with respect to the query vertex $u$ by our proposed algorithm. 
Let us point out that our algorithm can be easily modified so that
we only output high SimRank score vertices(because we just need to
set up the threshold to prune the similarity search).
We then compute the following value:
$$\frac{\textrm{\# of our high score vertices}}{\textrm{\# of the optimal high score vertices}}.$$

We also do the same thing for high score vertices computed by
Fogaras and R\'acz~\cite{fogaras2005scaling}(we used the same parameter $R'=100$ presented in~\cite{fogaras2005scaling}).
We perform this operation 100 times, and take the average.
The result is in Table~\ref{tbl:Accuracy}.
We can see that our algorithm actually gives very accurate results.
In addition, our algorithm gives better accuracy than Fogaras and R\'acz~\cite{fogaras2005scaling}.

\subsubsection{Comparison with existing results}
\label{sec:Comparison}

In this subsection,
we compare our algorithm with two state-of-the-art algorithms for computing SimRank, 
and show that our algorithm outperforms significantly in terms of scalability.

\paragraph{Comparison with the state-of-the-art all-pairs algorithm}

Yu et al.~\cite{yu2012space} proposed an efficient all-pairs algorithm;
the time complexity of their algorithm is $O(T n m)$,
and the space complexity is $O(n^2)$, where $T$ is the number of the iterations.
We implemented their algorithm and evaluated it in comparison with ours.
We used the same parameters presented in \cite{yu2012space}.

Results are shown in Table~\ref{tbl:result};
the omitted results (---) mean that their algorithm failed to allocate memory.
From the results, we observe that their algorithm is
a little faster than ours in query time,
but our algorithm uses much less space(15--30 times).
In fact, their algorithm failed for graphs with a million edges,
because of memory allocation.
More importantly, their algorithm cannot estimate the memory usage
before running the algorithm. 
Moreover, since our all-pairs algorithm can easily be parallelized to multiple machines, if there are 100 machines, even for graphs of
billions size, our all-pairs algorithm can output all top-20 vertices in less than 5 days. Thus, our algorithm significantly outperforms their algorithm in terms of scalability.

\paragraph{Comparison with the state-of-the-art single-pair and single-source algorithm}

Fogaras and R\'acz~\cite{fogaras2005scaling} proposed an efficient single-pair algorithm
that estimates SimRank scores with Monte Carlo simulation.
Like our approach, their algorithm also consists of two phases, a preprocessing phase and a query phase.
In the preprocessing phase, their algorithm generates $R'$ random walks and stores the walks efficiently;
this phase requires $O(n R')$ time and $O(n R')$ space.
The query phase phase requires $O(T n R')$ time, where $T$ is the number of iterations.
We implemented their algorithm and evaluated it in comparison with ours.
We used the same parameter $R'=100$ presented in~\cite{fogaras2005scaling}.

We can see that their algorithm is faster in query time. But we suspect that
this is due to relaxing accuracy, as in the previous subsection.
In order to obtain the same accuracy as our algorithm,
we suspect that $R'$ should be 500--1000, which implies that
their algorithm would be at least 5--10 times slower, and
require at leats 5--10 times more space.

In this case, their algorithm would fail for
graphs with more than ten millions edges because of memory allocation. Even for the
case $R'=100$, our algorithm uses much less space(10--20 times), and 
their algorithm failed for graphs with
more than 70 millions edges because of memory allocation.
Therefore we can conclude that our algorithm significantly outperforms their algorithm in terms of scalability.

\begin{table}[tb]
\tbl{Accuracy.  \label{tbl:Accuracy}}{
\small
\centering
\begin{tabular}{l|rrr} \hline
  Dataset    & Threshold & Proposed & Fogaras and R\'acz~\cite{fogaras2005scaling} \\ \hline
  ca-GrQc    & 0.04 & 0.98665 & 0.92329 \\
             & 0.05 & 0.98854 & 0.92467 \\
             & 0.06 & 0.99461 & 0.95225 \\
             & 0.07 & 0.99554 & 0.92881 \\
             \hline
  as20000102 & 0.04 & 0.97831 & 0.94643 \\
             & 0.05 & 0.98727 & 0.94783 \\
             & 0.06 & 0.99177 & 0.94713 \\
             & 0.07 & 0.99550 & 0.94760 \\
             \hline
  wiki-Vote  & 0.04 & 0.81862 & 0.93491 \\
             & 0.05 & 0.88629 & 0.93760 \\
             & 0.06 & 0.90801 & 0.94215 \\
             & 0.07 & 0.94785 & 0.97916 \\
             \hline
  ca-HepTh   & 0.04 & 0.97142 & 0.88964 \\
             & 0.05 & 0.98782 & 0.94354 \\
             & 0.06 & 0.99673 & 0.91848 \\
             & 0.07 & 0.99746 & 0.93647 \\
             \hline
\end{tabular}}
\end{table}

\begin{table}
\tbl{Preprocess time, query time and space for our proposed algorithm, \cite{fogaras2005scaling}, and \cite{yu2012space}. The single-pair and single-source results are the average of 10 trials;
we omitted results of the all-pairs computation of our proposed algorithm for large
networks; other omitted results (---) mean that the algorithms failed to allocate memory.  \label{tbl:result}}{
\small
\centering
\tabcolsep=3pt
\begin{tabular}{l|rrrr|rrr|rrr} \hline
  Dataset           & \multicolumn{4}{c|}{Proposed} & \multicolumn{3}{c|}{\cite{fogaras2005scaling}} & \multicolumn{2}{c}{\cite{yu2012space}}\\
                    & Preproc. & Query & AllPairs & Index & Preproc. & Query & Index & AllPairs & Memory \\  \hline
  ca-GrQc           & 1.5 s & 2.6 ms   & 13.5 s & 2.4 MB & 110 ms & 0.11 ms & 22.7 MB  & 2.97 s & 66 MB \\
  as20000102        & 1.6 s & 18 ms    & 115 s  & 3.3 MB & 81 ms & 1.3 ms & 28.0 MB     & 0.13 s & 51 MB  \\
  Wiki-Vote         & 1.9 s & 3.8 ms   & 26.9 s & 5.3 MB & 110 ms & 0.41 ms & 31.1 MB  & 8.74 s & 138 MB \\
  ca-HepTh          & 1.8 s & 3.3 ms   & 32.3 s & 4.5 MB & 253 ms & 0.44 ms & 43.3 MB  & 23.3 s & 312 MB \\
  email-Enron       & 7.8 s & 24 ms    & 864 s  & 21.6 MB & 949 ms & 1.1 ms & 158 MB   & 302 s & 3.45 GB \\
  soc-Epinions1     & 19.5 s & 44 ms   & 3335 s & 18.9 MB & 1.8 s & 1.4 ms & 332 MB   & 777 s & 6.91 GB \\
  soc-Slashdot0811  & 20.2 s & 53 ms   & 4081 s & 48.6 MB & 1.9 s & 1.2 ms & 341 MB   & 747 s & 7,34 GB \\
  soc-Slashdot0902  & 21.3 s & 55 ms   & 4515 s & 51.2 MB & 2.0 s & 1.3 ms & 363 MB   & 694 s & 7.21 GB \\
  email-EuAll       & 55.6 s & 226 ms  & --- & 103 MB & 3.6 s & 14 ms & 1.1 GB  & --- & --- \\
  web-Stanford      & 69.6 s & 103 ms  & --- & 141 MB & 10.4 s & 10 ms & 1.2 GB  & --- & --- \\
  web-NotreDame     & 60.6 s & 73 ms   & --- & 163 MB & 7.6 s & 2.8 ms & 1.4 GB  & --- & --- \\
  web-BerkStan      & 240.2 s & 93 ms  & --- & 288 MB & 47.4 s & 4.3 ms & 3.8 GB  & --- & --- \\
  web-Google        & 156.7 s & 199 ms & --- & 211 MB & 24.0 s & 13 ms & 3.0 GB  & --- & --- \\
  dblp-2011         & 82.2 s & 16 ms   & --- & 144 MB & 16.4 s & 1.3 ms & 1.4 GB  & --- & --- \\
  in-2004           & 292.9 s & 95 ms  & --- & 451 MB & 46.4 s & 6.8 ms & 6.0 GB  & --- & --- \\
  flickr            & 622.7 s & 1.5 s  & --- & 513 MB & 90.1 s & 7.6 ms & 7.4 GB  & --- & --- \\
  soc-LiveJournal   & 2335.9 s & 431 ms& --- & 1.2 GB & 397.9 s & 27 ms & 21.6 GB  & --- & --- \\
  indochina-2004    & 1585.2 s & 714 ms& --- & 2.1 GB & --- & --- & ---  & --- & --- \\
  it-2004           & 3.1 h & 2.3 s    & --- & 11.2 GB & --- & --- & ---  & --- & --- \\
  twitter-2010      & 7.7 h & 17.4 s   & --- & 8.4 GB & --- & ---  & --- & ---  & --- & --- \\
  \hline
\end{tabular}}
\end{table}

\clearpage

\section{SimRank join$^1$}
\label{sec:join}
\footnotetext[1]{This section is based on our ICDE'15 paper~\cite{maehara2015scalable}}

\subsection{Motivation and overview}

Finally, in this section, we describe the SimRank join problem, which is formulated as follows.
\begin{problem}[SimRank join]
Given a directed graph $G = (V, E)$ and a threshold $\theta \in [0,1]$,
find all pairs of vertices $(i,j) \in V \times V$
for which the SimRank score of $(i, j)$ is greater than the threshold, i.e., $s(i,j) \ge \theta$,
\end{problem}

This problem is useful in the \emph{near-duplication detection problem}~\cite{chen2002origin,arasu2009large}.
Let us consider the World Wide Web,
which contains many ``very similar pages.''
These pages are produced by activities such as
file backup, caching, spamming, and authors moving to different institutions.
Clearly, these very similar pages are not desirable for data mining,
and should be isolated from the useful pages by near-duplication detection algorithms.

Near duplication detection problem is solved by the \emph{similarity join query}.
Let $s(i,j)$ be a similarity measure, i.e., for two objects $i$ and $j$, they are (considered as) similar if and only if $s(i,j)$ is large.
The \emph{similarity join query with respect to $s$} finds all pairs of objects $(i,j)$
with similarity score $s(i,j)$ exceeding some specified threshold $\theta$~\cite{white1996similarity}.
\footnote{Our version of the similarity join problem is sometimes called \emph{self-similarity join}. Some authors have defined ``similarity join'' as follows: given two sets $S$ and $T$, find all similar pairs $(i,j)$ where $i \in S$ and $j \in T$. In this paper, we consider only the self-similarity join query.}%
\footnote{There is also a \emph{top-$k$} version of the similarity join problem that enumerates the $k$ most similar pairs. In this paper, we consider only the threshold version of the similarity join problem.}
The similarity join is a fundamental query for a database,
and is used in applications, such as
merge/purge~\cite{hernandez1995merge},
record linkage~\cite{fellegi1969theory},
object matching~\cite{sivic2003video}, and
reference reconciliation~\cite{dong2005reference}.

Selecting the similarity measure $s(i,j)$ is an important component of the similarity join problem.
Similarity measures on graphs have been extensively investigated.
Here, we are interested in \emph{link-based similarity measures},
which are determined by sorely the link structure of the network.
For applications in the World Wide Web,
by comparing content-based similarity measures, which are determined by the content data stored on vertices (e.g., text and images),
link-based similarity measures are more robust against spam pages and/or machine-translated pages.

\subsubsection{Difficulty of the problem}

In solving the SimRank join problem,
the following obstacles must be overcome.
\begin{enumerate}
  \item There are many similar-pair candidates.
  \item Computationally, SimRank is very expensive.
\end{enumerate}
To clarify these issues, we compare the SimRank with the Jaccard similarity,
where the \emph{Jaccard similarity} between two vertices $i$ and $j$ is given by
\begin{align*}
  J(i,j) := \frac{|\In(i) \cap \In(j)|}{|\In(i) \cup \In(j)|}.
\end{align*}

Regarding the first issue,
since the Jaccard similarity satisfies $J(i,j) = 0$ for all pairs of vertices $(i,j)$ with $d(i,j) \ge 3$ (i.e., their distance is at least three),
the number of possibly similar pairs (imposing the Jaccard similarity) is easily reduced to much smaller than all possible pairs.
This simple but fundamental concept is adopted in
many existing similarity join algorithms~\cite{sarawagi2004efficient}.
However, since the SimRank exploits the information in multihop neighborhoods and hence scans the entire graph,
it must consider all $O(n^2)$ pairs, where $n$ is the number of vertices.
Therefore, whereas the Jaccard similarity is adopted in ``local searching,'' the SimRank similarity must look at the global influence of all vertices, which requires tracking of all $O(n^2)$ pairs.

Regarding the second issue,
the Jaccard similarity of two vertices can be very efficiently computed (e.g., in $O(|\In(i)| + |\In(j)|)$ time using a straightforward method or in $O(1)$ time using MinHash~\cite{broder1997resemblance,lee2011similarity}).
Conversely, SimRank computation is very expensive.

Notably, until recently, most SimRank algorithms have computed the all-pairs SimRank scores~\cite{jeh2002simrank,lizorkin2010accuracy,yu2010taming}, which requires at least $O(n^2)$ time,
and if we have the all-pairs SimRank scores, the SimRank join problem is solved in $O(n^2)$ time.
For example, in one investigation of the SimRank join problem~\cite{zheng2013efficient}, the all-pairs SimRank scores were first computed by an existing algorithm and an index for the SimRank join was then constructed.
Therefore, developing scalable algorithm for the SimRank join problem is a newer challenging problem.

\subsubsection{Contribution and overview}

Here, we propose a scalable algorithm for the SimRank join problem, and perform experiments on large real datasets.
The computational cost of the proposed algorithm only depends on the number of similar pairs, but does not depend on all pairs $O(n^2)$.
The proposed algorithm scales up to the network of 5M vertices and 70M edges.
By comparing with the state-of-the-art algorithms,
it is about 10 times faster, and requires about only 10 times smaller memory.

This section overviews our algorithm, which consists of two phases:
\emph{filter} and \emph{verification}.
The former enumerates the similar pair candidates,
and then the latter decides whether each candidate pair is actually similar.
Note that this framework is commonly adopted in similarity join algorithms~\cite{xiao2011efficient,deng2014massjoin}.
A more precise description of the two phases is given below.

The filter phase is the most important phase of the proposed algorithm
because it must overcome both (1) and (2) difficulties, discussed in previous subsection.
We combines the following three techniques for this phase.
The details are discussed in Section~\ref{sec:filter_top}.
\begin{enumerate}
  \item
    We adopt the \emph{SimRank linearization} (Section~\ref{sec:linearization}) by which the SimRank is computed as a solution to a linear equation.
  \item We solve the linear equation approximately by the \emph{Gauss-Southwell algorithm}~\cite{southwell1940relaxation,southwell1946relaxation},
    which avoids the need to compute the SimRank scores for non-similar pairs (Sections~\ref{sec:gausssouthwell},\ref{sec:filter}).
  \item We adopt the \emph{stochastic thresholding} to reduce the memory used in the Gauss-Southwell algorithm (Section~\ref{sec:thresholding}).
\end{enumerate}

The verification phase is simpler than the filter phase.
We adopt the following two techniques for this phase.
The details are discussed in Section~\ref{sec:verification}.
\begin{enumerate}
  \item
    We run a Monte-Carlo algorithm for each candidate to decide whether the candidate is actually similar.
    This can be performed in parallel.
  \item
    We control the number of Monte-Carlo samples adaptively to reduce the computational time.
\end{enumerate}
It should be emphasized that, we give theoretical guarantees
for all techniques used in the algorithm.
All omitted proofs are given in Appendix.

The proposed algorithm is evaluated by experiments on real datasets (Section~\ref{sec:experiments}).
The algorithm is 10 times faster, and requires only 10 times smaller memory that of the state-of-the-arts algorithms,
and scales up to the network of 5M vertices and 70M edges.
Since the existing study~\cite{zheng2013efficient} only performed in 338K vertices and 1045K edges, our experiment scales up to the 10 times larger network.
Also, we empirically show that the all techniques used in the algorithm works effectively.

We also verified that 
the distribution of the SimRank scores on a real-world network
follows a power-law distribution~\cite{cai2009efficient},
which has been verified only on small networks. 

\subsection{Filter phase}
\label{sec:filter_top}

\begin{algorithm}[tb]
\caption{SimRank join algorithm.}
\label{alg:simrankjoin}
\begin{algorithmic}[1]
\Procedure{SimRankJoin}{$\theta$}
\State{Compute two sets $J_L$ and $J_H$.}
\State{Output all $(i,j) \in J_L$.}
\For{$(i,j) \in J_H \setminus J_L$}
\If{$i$ and $j$ are really similar}
\State{Output $(i,j)$.}
\EndIf
\EndFor
\EndProcedure
\end{algorithmic}
\end{algorithm}

In this section, we discuss the details of the filter phase
for enumerating similar-pair candidates.
We will discuss the details of the verification phase in the next section. 
Let us give overview of the filter phase. 

Let $J(\theta) = \{ (i,j) : s(i,j) \ge \theta \}$ be the set of similar pairs.
The filter phase produces two subsets $J_L(\theta,\gamma)$ and $J_H(\theta,\gamma)$ such that
\begin{align}
\label{eq:JLJH}
  J_L(\theta,\gamma) \subseteq J(\theta) \subseteq J_H(\theta,\gamma),
\end{align}
where $\gamma \in [0,1]$ is an accuracy parameter.
Here,
$J_L(\theta,\gamma)$ is monotone increasing in $\gamma$,
$J_H(\theta, \gamma)$ is monotone decreasing in $\gamma$,
and $J_L(\theta,1) = J(\theta) = J_H(\theta,1)$.
%
%
Note that, our filter phase gives 100\%-precision solution $J_L(\theta,\gamma)$ and 100\%-recall solution $J_H(\theta,\gamma)$.
These two sets might be useful in some applications.

Let us consider how to implement the filter phase.
The basic idea is that ``compute only relevant entries of SimRank.''
In order to achieve this idea, we combine the two techniques,
the \emph{linearization of the SimRank}~\cite{kusumoto2014scalable,maehara2014efficient},
and
the \emph{Gauss-Southwell algorithm}~\cite{southwell1940relaxation,southwell1946relaxation} for solving a linear equation.
The linearization of the SimRank
is a technique to convert a SimRank problem to a linear equation problem.
By using this technique,
the problem of computing large entries of SimRank
is transformed to the problem of computing large entries
of a solution of a linear equation.
To solve this linear algebraic problem,
we can use the Guass-Southwell algorithm.
By error analysis of the Guass-Southwell algorithm (given later in this paper),
our filter algorithm obtains the lower and the upper bound of the SimRank of each pair $(i,j)$.
Using these bounds, the desired sets $J_L(\theta,\gamma)$ and $J_H(\theta,\gamma)$ are obtained, where $\gamma$ is an accuracy parameter used in the Guass-Southwell algorithm.

The above procedure is already much efficient;
however, we want to scale up the algorithm for large networks.
The bottleneck of the above procedure is the memory allocation.
We resolve this issue by introducing the \emph{stochastic thresholding} technique.

This is the overview of the proposed filter phase.
In the following subsections, we give details of the filter phase.

\subsubsection{Gauss-Southwell algorithm}
\label{sec:gausssouthwell}

By the linearization of the SimRank, we obtained the linear equation \eqref{eq:linearizedsimrank}.
We want to solve this linear equation in $S$;
however, since it has $O(n^2)$ variables, we must keep track only on large entries of $S$.
For this purpose, we adopt the Gauss-Southwell algorithm~\cite{southwell1940relaxation,southwell1946relaxation},
which is a very classical algorithm for solving a linear equation.
In this subsection, we describe the Gauss-Southwell algorithm for a general linear equation $A x = b$,
and in the next subsection, we apply this method for the linearized SimRank equation \eqref{eq:linearizedsimrank}.

Suppose we desire an approximate solution to the linear system $A x = b$,
where $A$ is an $n \times n$ matrix with unit diagonal entries,
i.e., $A_{ii} = 1$ for all $i = 1, \ldots, n$.
Let $\epsilon > 0$ be an accuracy parameter.
The Gauss-Southwell algorithm is an iterative algorithm.
Let $x^{(t)}$ be the $t$-th solution
and $r^{(t)} := b - A x^{(t)}$ be the corresponding residual.
At each step,
the algorithm chooses an index $i$ such that $|r^{(t)}_i| \ge \epsilon$,
and updates the solution as
\begin{align}
  x^{(t+1)} = x^{(t)} + r^{(t)}_i e_i,
\end{align}
where $e_i$ denotes the $i$-th unit vector.
The corresponding residual becomes
\begin{align}
  r^{(t+1)} = b - A x^{(t+1)} = r^{(t)} - r^{(t)}_i A e_i.
\end{align}
Since $A$ has unit diagonals, the $i$-th entry of $r^{(t+1)}$ is zero.
Repeating this process
until $r^{(t)}_i < \epsilon$ for all $i = 1, \ldots, n$,
we obtain a solution $x$ such that $\| b - A x \|_\infty < \epsilon$.
This algorithm is called the \emph{Gauss-Southwell} algorithm.

Note that this algorithm is recently rediscovered
and applied to the personalized PageRank computation.
See Section~\ref{sec:relatedwork} for related work.

\subsubsection{Filter phase}
\label{sec:filter}

We now propose our filtering algorithm.
First, we compute the diagonal correction matrix $D$ by Algorithm~\ref{alg:DiagonalEstimation},
reducing the SimRank computation problem to the linear equation~\eqref{eq:linearizedsimrank}.
The Gauss-Southwell algorithm is then applied to the equation.

The $t$-th solution $S^{(t)}$ and the corresponding residual $R^{(t)}$ are maintained such that
\begin{align}
  D - (S^{(t)} - c P^\top S^{(t)} P) = R^{(t)}.
\end{align}
Initial conditions are $S^{(0)} = O$ and $R^{(0)} = D$.
At each step, the algorithm finds an entry $(i,j)$ such that $|R^{(t)}_{ij}| \ge \epsilon$,
and then updates the current solution as
\begin{align}
  S^{(t+1)} = S^{(t)} + R^{(t)}_{ij} e_i e_j^\top.
\end{align}
The corresponding residual becomes
\begin{align}
  \label{eq:filteriter}
  R^{(t+1)} = R^{(t)} - R^{(t)}_{ij} e_i e_j^\top + c R^{(t)}_{ij} (P^\top e_i) (P^\top e_j)^\top.
\end{align}
Since we have assumed that $G$ is simple, the $i$-th entry of $P^\top e_i$ is zero,
the $(i,j)$-th entry of $R^{(t+1)}$ is also zero.
The algorithm repeats this process until
$|R^{(t)}_{ij}| < \epsilon$ for all $i,j \in V$.
The procedure is outlined in Algorithm~\ref{alg:gausssouthwell}.

\begin{algorithm}[tb]
\caption{Gauss-Southwell algorithm used in Algorithm~\ref{alg:filter}.}
\label{alg:gausssouthwell}
\begin{algorithmic}[1]
\Procedure{GaussSouthwell}{$\epsilon$}
  \State{$S^{(0)} = O, \, R^{(0)} = D, \, t = 0$}
  \While{there is $(i,j)$ such that $|R^{(t)}_{ij}| > \epsilon$}
  \State{$S^{(t+1)} = S^{(t)} + R^{(t)}_{ij} e_i e_j^\top$}
  \State{$R^{(t+1)} = R^{(t)} - R^{(t)}_{ij} e_i e_j^\top + c R^{(t)}_{ij} (P^\top e_i) (P^\top e_j)^\top$}
  \State{$t \leftarrow t + 1$}
  \EndWhile
  \State{Return $S^{(t)}$ as an approximate SimRank.}
\EndProcedure
\end{algorithmic}
\end{algorithm}

We first show the finite convergence of the algorithm, whose proof will be given in Appendix.
\begin{proposition}
\label{prop:terminate}
Algorithm~\ref{alg:gausssouthwell} terminates at most $t = \Sigma / \epsilon$ steps, where $\Sigma = \sum_{ij} S_{ij}$ is the sum of all SimRank scores.
\end{proposition}
Since the $(i,j)$ step is performed in $O(|\In(i)||\In(j)|)$ time
with $O(|\In(i)||\In(j)|)$ memory allocation,
the overall time and space complexity is $O(I_{\mathrm{max}}^2 \Sigma/\epsilon)$,
where $I_{\mathrm{max}}$ is the maximum in-degree of $G$.

We now show that Algorithm 6 guarantees an approximate solution (whose proof will be given in Appendix).
\begin{proposition}
\label{prop:errorestimate}
Let $\gamma \in [0,1)$ be an accuracy parameter and
$\tilde S$ be the approximate SimRank obtained by
Algorithm~\ref{alg:gausssouthwell} with $\epsilon = (1 - c) (1 - \gamma) \theta$.
Then we have
\begin{align}
  \label{eq:errorestimate}
  0 \le S_{ij} - \tilde S_{ij} \le (1 - \gamma) \theta
\end{align}
for all $i, j \in V$.
\end{proposition}
The left inequality
states if $\tilde S_{ij} \ge \theta$, $S_{ij} \ge \theta$.
Similarly, the right inequality states that
if $S_{ij} \ge \theta$, $\tilde S_{ij} \ge \gamma \theta$.
Thus, letting
\begin{align}
  J_L(\theta,\gamma) &:= \{ (i,j) : \tilde S_{ij} \ge \theta \}, \\
  J_H(\theta,\gamma) &:= \{ (i,j) : \tilde S_{ij} \ge \gamma \theta \},
\end{align}
we obtain \eqref{eq:JLJH}.

$J_L$ and $J_H$ can be accurately estimated by letting $\gamma \to 1$.
However, since the complexity is proportional to $O(1/\epsilon) = O(1/(1 - \gamma))$, a large $\gamma$ is precluded.
Therefore, in practice, we set to some small value (e.g., $\gamma = 0$)
and
verify whether the pairs $(i,j) \in J_H(\theta,\gamma) \setminus J_L(\theta,\gamma)$ are actually similar by an alternative algorithm.

\begin{algorithm}[tb]
\caption{Filter procedure.}
\label{alg:filter}
\begin{algorithmic}[1]
\Procedure{Filter}{$\theta$, $\gamma$}
\State{Compute diagonal correction matrix $D$.}
\State{Compute approximate solution $\tilde S$ of linear equation $S = c P^\top S P + D$ by Gauss-Southwell algorithm with accuracy $\epsilon = (1-c) \gamma \theta$.}
\State{Output $J_L = \{ (i,j) : \tilde S_{ij} \ge \theta \}$ and $J_H = \{ (i,j) : \tilde S_{ij} \ge \gamma \theta \}$.}
\EndProcedure
\end{algorithmic}
\end{algorithm}

\subsubsection{Stochastic thresholding for reducing memory}
\label{sec:thresholding}

We cannot predict the required memory before running the algorithm (which depends on the SimRank distribution);
therefore the memory allocation is a real bottleneck
in the filter procedure (Algorithm~\ref{alg:filter}).
To scale up the procedure, we must reduce the waste of memory.
Here, we develop a technique that reduces the space complexity.

We observe that, in Algorithm~\ref{alg:filter}, some entries $R_{ij}$ are only used to store the values,
and never used in future because they do not exceed $\epsilon$;
therefore we want to reduce storing of these values.
Here, a thresholding technique,
which skips memory allocations for very small values, seems effective.
This kind of heuristics is frequently used in SimRank computations~\cite{cai2009efficient,lizorkin2010accuracy,yu2010taming}.
However, a simple thresholding technique may cause large errors because it ignores the accumulation of small values.

To guarantee the theoretical correctness of this heuristics,
we use the \emph{stochastic thresholding} instead of the deterministic thresholding.
When the algorithm requires to allocate a memory,
it skip the allocation with some probability depending on the value (a small value should be skipped with high probability).
Intuitively, if the value is very small, the allocation may be skipped.
However, if there are many small values, one of them may be allocated,
and hence the error is bounded stochastically. The following proposition shows this fact (whose proof will be given in Appendix).

\begin{proposition}
\label{prop:thresholding}
Let $a_1, a_2, \ldots$ be a nonnegative sequence and let $A = \sum_{i=1}^\infty a_i < \infty$.
Let $\beta > 0$.
Let $\tilde A$ be the sum of these values with the stochastic thresholding where the skip probability is $p(a_i) = \min\{1, \beta a_i\}$.
Then we have
\begin{align}
  \P \{ A - \tilde A \ge \delta \} \le \exp(-\beta \delta).
\end{align}
\end{proposition}

We implement the stochastic thresholding (Algorithm~\ref{alg:thresholding}) in Gauss-Southwell algorithm (Algorithm~\ref{alg:gausssouthwell}).
Then, theoretical guarantee in Proposition~\ref{prop:errorestimate} is modified as follows.
\begin{proposition}
\label{prop:memory}
Let $\gamma \in [0,1)$ be an accuracy parameter, $\beta$ be a skip parameter,
and $\tilde S$ be the approximate SimRank obtained by
Algorithm~\ref{alg:gausssouthwell} using stochastic thresholding with $\epsilon = (1 - c) (1 - \gamma) \theta$.
Then we have $ 0 \le S_{ij} - \tilde S_{ij}$ and
\begin{align}
  \P\{ S_{ij} - \tilde S_{ij} \le (1 - \gamma) \theta + \delta \} \le \exp(-(1-c) \beta \delta)
\end{align}
for all $i, j \in V$.
\end{proposition}
Thus, by letting $\beta \propto 1/\delta$, we can reduce the misclassification probability arbitrary small.
We experimentally evaluate the effect of this technique
in Section~\ref{sec:experiments}

\begin{algorithm}[tb]
\caption{Stochastic thresholding for $R_{ij} \leftarrow R_{ij} + a$.}
\label{alg:thresholding}
\begin{algorithmic}[1]
  \Procedure{StochasticThresholding}{$R$,$i$,$j$,$a$;$\beta$}
  \If{$R_{ij}$ is already allocated}
  \State{$R_{ij} \leftarrow R_{ij} + a$.}
  \Else
  \State{draw a random number $r \in [0,1]$.}
  \If{$r > \beta a$}
  \State{skip allocation.}
  \Else
  \State{allocate memory for $R_{ij}$ and store $R_{ij} = a$.}
  \EndIf
  \EndIf
  \EndProcedure
\end{algorithmic}
\end{algorithm}

\subsection{Verification phase}
\label{sec:verification}

In the previous section, we described the filter phase for enumerating the similar-pair candidates.
In this section, we discuss the verification phase for deciding whether each candidate $(i,j)$ is actually similar.
It should be mentioned that this procedure for each pair can be performed in parallel, i.e.,
if there are $M$ machines, the computational time is reduced to $1/M$.

Our verification algorithm is a \emph{Monte-Carlo algorithm}
based on the representation of the first meeting time~\eqref{eq:randomsurferpair}.
$R$ samples of the first meeting time $\tau^{(1)}(i,j), \ldots, \tau^{(R)}(i,j)$ are obtained from $R$ random walks
and the SimRank score is then estimated by the sample average:
\begin{align}
\label{eq:sampleaverage}
  s(i,j) \approx s^{(R)}(i,j) := \frac{1}{R} \sum_{k=1}^R c^{\tau^{(k)}(i,j)}.
\end{align}
To accurately estimate $s(i,j)$ by \eqref{eq:sampleaverage},
many samples (i.e., $r \ge 1000$) are required~\cite{fogaras2005scaling}.
However, to decide whether $s(i,j)$ is greater than or smaller than the threshold $\theta$, much fewer samples are required.
\begin{proposition}
\label{prop:hoeffding}
Let $\delta^{(R)} = |s^{(R)}(i,j) - \theta|$.
If $s(i,j) \ge \theta$ then
\begin{align}
  \P\{ s^{(R)}(i,j) < \theta \} \le \exp\left( - 2 R \delta^{(R) 2} \left(\frac{1-c}{c}\right)^2 \right).
\end{align}
Similarly, if $s(i,j) \le \theta$ then
\begin{align}
  \P\{ s^{(R)}(i,j) > \theta \} \le \exp\left( - 2 R \delta^{(R) 2} \left(\frac{1-c}{c}\right)^2 \right).
\end{align}
\end{proposition}
The proof will be given in Appendix. 
This shows that, if $s(i,j)$ is far from the threshold $\theta$,
it can be decided with a small number of Monte-Carlo samples.

Now we describe our algorithm.
We optimize the number of Monte-Carlo samples,
by adaptively increasing the samples.
Starting from $R = 1$, our verification algorithm loops through the following procedure.
The algorithm obtains $s^{(R)}(i,j)$ by performing
two random walks and computing $\tau^{(R)}(i,j)$.
It then checks the condition
\begin{align}
\label{eq:termination}
  R \delta^{(R) 2} \ge \frac{\log(1/p)}{2} \left( \frac{c}{1-c} \right)^2,
\end{align}
where $p \in (0,1)$ is a tolerance probability for misclassification.
If this condition is satisfied, we determine $s(i,j) < \theta$ or $s(i,j) > \theta$,
The misclassification probability of this decision is at most $p$
by Proposition~\ref{prop:hoeffding}.

\begin{algorithm}[tb]
\caption{Verification procedure.}
\label{alg:verification}
\begin{algorithmic}[1]
\Procedure{Verification}{$i$, $j$; $\theta$, $p$, $R_\mathrm{max}$}
\For{$R = 1, \ldots, R_{\mathrm{max}}$}
\State{Perform two independent random walks to obtain the first meeting time $\tau^{(r)}(i,j)$.}
\State{$s^{(R)} = (1/R) \sum_{k=1}^R \tau^{(k)}(i,j)$.}
\State{$\delta^{(R)} = |s^{(R)} - \theta|$.}
\If{$R \delta^{(R) 2}$ is large, i.e., \eqref{eq:termination} holds}
\State{\textbf{break}}
\EndIf
\EndFor
\If{$s^{(R)} \le \theta$}
\State{\Return $s(i,j) \ge \theta$}
\Else
\State{\Return $s(i,j) \le \theta$}
\EndIf
\EndProcedure
\end{algorithmic}
\end{algorithm}

We now evaluate the number of samples required in this procedure.
By taking expectation of \eqref{eq:termination}, we obtain
the expected number of iterations $R_{\mathrm{end}}$ as
\begin{align}
  \E[ R_\textrm{end} ] \ge \frac{1}{(s(i,j) - \theta)^2} \frac{\log(1/p)}{2} \left( \frac{c}{1-c} \right)^2,
\end{align}
This shows that the number of samples quadratically depends on the inverse of the difference between the score $s(i,j)$ and the threshold $\theta$.
In practice, we set an upper bound $R_\mathrm{max}$ of $R$, and
terminate the iterations after $R_\mathrm{max}$ steps.
Our verification algorithm is shown in Algorithm~\ref{alg:verification}.


\subsection{Experiments}
\label{sec:experiments}

In this section, we perform numerical experiments to evaluate the proposed algorithm.

\begin{table}[tb]
\tbl{Scalability of the proposed algorithm.  \label{tbl:scalability}}{
\centering
\begin{tabular}{l|rr|rrrrrr} \hline
  dataset & $|V|$ & $|E|$ & $|J_L|$ & $|J_H|$ & estimate & filter & verification & memory \\ \hline
  amazon0302 & 261,111 & 1,234,877 & 1,059 & 1,338,677 & 3,355 & 5.3 s  & 5.3 s & 505 MB \\
  amazon0312 & 400,727 & 3,200,440 & 3,229 & 1,579,331 & 6,563 & 40.8 s & 4.8 s & 2.6 GB \\
  amazon0505 & 410,236 & 3,356,824 & 3,898 & 1,174,540 & 7,951 & 15.6 s & 3.8 s & 1.8 GB \\
  amazon0601 & 403,394 & 3,387,388 & 5,139 & 1,289,190 &10,375 & 16.3 s & 4.3 s & 1.8 GB \\
  as-caida   & 26,475 & 106,762 & 2,141,694 & 9,281,251 & 2,601,620 & 6.1 s & 84.0 s & 776 MB \\
  as-skitter & 1,696,416 & 11,095,298 & 3,386,713 & 23,240,912 & 4,327,637 & 223.3 s & 54.7 s & 3.8 GB \\
  as20000102 & 6,474 & 13,895 & 276,586 & 1,309,773 & 378,312 & 0.8 s & 9.2 s & 108 MB \\
  ca-AstroPh & 18,772 & 396,160 & 1,331 & 31,257 & 2,552 & 1.3 s & 0.4 s & 44 MB \\
  ca-CondMat & 23,133 & 186,936 & 3,100 & 75,607 & 5,969 & 0.2 s & 0.5 s & 32 MB \\
  ca-GrQc    & 5,242 & 28,980 & 1,497 & 17,620 & 2,696 & 1.4 s & 0.1 s & 3 MB \\
  ca-HepPh   & 12,008 & 237,010 & 1,948 & 32,436 & 3,446 & 1.4 s & 0.4 s & 19 MB \\
  ca-HepTh   & 9,877 & 51,971 & 2,461 & 32,750 & 4,349 & 0.05 s & 0.1 s & 8 MB \\
  cit-HepPh  & 34,546 & 421,578 & 7,827 & 218,291 & 10,720 & 9.5 s & 0.5 s & 262 MB \\
  cit-HepTh  & 27,770 & 352,807 & 5,441 & 126,057 & 5,441 & 4.6 s & 0.3 s & 250 MB \\
  cit-Patents & 3,774,768 & 16,518,948 & 151,096 & 4,406,829 & 209,793 & 42.7 s & 7.1 s & 3.6 GB \\
  com-amazon & 334,863 & 925,872 & 96,257 & 2,193,701 & 164,967 & 4.0 s & 14.3  & 353 MB \\
  com-dblp   & 317,080 & 1,049,866 & 136,176 & 1,862,027 & 222,853 & 4.2 s & 14.0 s & 448 MB \\
  email-Enron & 36,692 & 367,662 & 1,204,942 & 2,919,124 & 1,348,743 & 4.1 s & 14.3 s & 470 MB \\
  email-EuAll & 265,214 & 420,045 & 131,109,661 & 151,333,618  & 133,693,685 & 84.5 s & 116.4 s & 10.6 GB \\
  p2p-Gnutella04 & 10,876 & 39,994 & 3,053 & 32,226 & 3,970 & 0.1 s & 0.2 s & 14 MB \\
  p2p-Gnutella05 & 8,846 & 31,839 & 2,527 & 27,101 & 3,352 & 0.06 s & 0.1 s & 12 MB \\
  p2p-Gnutella06 & 8,717 & 31,525 & 2,902 & 26,639 & 3,691 & 0.07 s & 0.2 s & 11 MB \\
  p2p-Gnutella08 & 6,301 & 20,777 & 3,582 & 25,772 & 4,521 & 0.05 s & 0.1 s & 8 MB \\
  p2p-Gnutella09 & 8,114 & 26,013 & 5,091 & 34,142 & 6,234 & 0.06 s & 0.3 s & 10 MB \\
  p2p-Gnutella24 & 26,518 & 65,369 & 23,894 & 131,939 & 30,072 & 0.2 s & 0.8 s & 26 MB \\
  p2p-Gnutella25 & 22,687 & 54,705 & 21,893 & 114,388 & 26,939 & 0.1 s & 1.0 s & 20 MB \\
  p2p-Gnutella30 & 36,682 & 88,328 & 41,164 & 192,307 & 49,644 & 0.5 s & 1.4 s & 33 MB \\
  p2p-Gnutella31 & 62,586 & 147,892 & 71,054 & 334,950 & 86,367 & 0.6 s & 2.1 s & 57 MB \\
  soc-Epinions1  & 75,879 & 508,837 & 205,650 & 1,201,677 & 252,012 & 6.0 s & 7.9 s & 462 MB \\
  soc-LiveJournal & 4,847,571 & 68,993,773 & 3,098,597 & 30,715,479 & 3,975,507 & 640.4 s & 347.8 s & 23 GB \\
  soc-Slashdot0811 & 77,360 & 905,468 & 2,126 & 1,099,584 & 15,466 & 7.1 s & 9.7 s & 639 MB \\
  soc-Slashdot0902 & 82,168 & 948,464 & 1,904 & 1,072,625 & 11,447 & 7.4 s & 9.2 s & 627 MB \\
  soc-pokec & 1,632,803 & 30,622,564 & 314,476 & 3,739,302 & 406,539 & 134.4 s & 43.1 s & 6.5 GB \\
  web-BerkStan  & 685,230 & 7,600,595 & --- & --- & --- & --- & --- & --- \\
  web-Google    & 875,713 & 5,105,039 & 8,054,397 & 100,247,737 & 11,416,471 & 262.9 s & 361.0 s & 24.7 GB \\
  web-NotreDame & 325,729 & 1,497,134 & 10,206,009 & 78,006,949 & 10,895,990 & 77.7 s & 188.5 s & 9.3 GB \\
  web-Stanford  & 281,903 & 2,312,497 & 21,438,881 & 427,602,788 & 25,890,022 & 1519.6 s & 1688.2 s & 61.3 GB \\
  wiki-Talk & 2,394,386 & 5,021,410 & 2,647,967 & 6,703,468 & 2,965,752 & 23.4 s & 16.5 s & 1.2 GB \\
  wiki-Vote & 7,115 & 103,589 & 10,420 & 54,613 & 12,581 & 0.3 s & 0.1 s & 26 MB \\
  \hline
\end{tabular}

}
\end{table}

We used the datasets shown in Table~\ref{tbl:scalability}.
These are obtained by Stanford Large Network Dataset Collection\footnote{\textit{https://snap.stanford.edu/data/}}.

All experiments were conducted on an Intel Xeon E5-2690
2.90GHz CPU (32 cores) with 256GB memory running Ubuntu 12.04.
The proposed algorithm was implemented in C++ and was compiled
using g++v4.6 with the -O3 option.
We have used OpenMP for parallel implementation.

\subsubsection{Scalability}

We first show that the proposed algorithm is scalable.
For real datasets,
we compute the number of pairs that has the SimRank score greater than $\theta = 0.2$,
where the parameters of the algorithm are set to the following.
\begin{itemize}
  \item For the filter phase, accuracy parameter for the Gauss-Southwell algorithm (Algorithm~\ref{alg:gausssouthwell}) is set to $\gamma = 0$, and the skip parameter $\beta$ for the stochastic thresholding (Algorithm~\ref{alg:thresholding}) is set to $\beta = 100$.
  \item For the verification phase, the tolerance probability $p$ is set to $p = 0.01$ and the maximum number of Monte-Carlo samples is set to $R_{\mathrm{max}} = 1000$.
\end{itemize}
The result is shown in Table~\ref{tbl:scalability},
which is the main result of this paper.
Here, ``dataset'', ``$|V|$'', and ``$|E|$'' denote the statistics of dataset,
``$|J_L|$'' and ``$|J_H|$'' denote the size of $J_L$ and $J_H$ in \eqref{eq:JLJH},
``estimate'' denotes the number of similar pairs estimated by the algorithm,
``filter'' and ``verification'' denote the real time needed to compute filter and verification step,
and ``memory'' denotes the allocated memory during the algorithm.

This shows the proposed algorithm is very scalable.
In fact, it can find the SimRank join for the networks of
5M vertices and 68M edges (``soc-LiveJournal'') in a 1000 seconds with 23 GB memory.

Let us look into the details.
The computational time and the allocated memory depend on the number of the similar pairs,
and even if the size of two networks are similar, the number of similar pairs in these networks can be very different.
For example, the largest instance examined in the experiment is
``soc-LiveJournal'', which has 5M vertices and 68M edges.
However, since it has a small number of similar pairs,
we can compute the SimRank join.
On the other hand, ``web-BerkStan'' dataset, which has only 0.6M vertices and 7M edges,
we cannot compute the SimRank join because it has too many similar pairs.

\subsubsection{Accuracy}

Next, we verify the accuracy (and the correctness) of the proposed algorithm.
We first compute the exact SimRank scores $S^*$ by using the original SimRank algorithm. 
Then, we compare the exact scores with the solution $S$ obtained by the proposed algorithm.
Here, the accuracy of the solution is measured by the precision, the recall, and the F-score~\cite{baeza1999modern}, defined by
\begin{align*}
  \text{precision} = \frac{|S \cap S^*|}{|S|}, \
  \text{recall}    = \frac{|S \cap S^*|}{|S^*|}, \
  \text{F}         = \frac{2 |S \cap S^*|}{|S|+|S^*|}.
\end{align*}
We use the same parameters as the previous subsection.
Since all-pairs SimRank computation is expensive,
we used relatively small datasets for this experiment.

The result is shown in Table~\ref{tbl:accuracy}.
This shows the proposed algorithm is very accurate;
the obtained solutions have about $\text{precision} \approx 97\%$, $\text{recall} \approx 92\%$, and $\text{F-score} \approx 95\%$.
Since the precision is higher than the recall,
the algorithm produces really similar pairs.


\subsubsection{Comparison with the state-of-the-arts algorithms}

We then compare the proposed algorithm with some state-of-the-arts algorithms.
For the proposed algorithm, we used the same parameters described in the previous section.
For the state-of-the-arts algorithms, we implemented the following two algorithms%
\footnote{
We have also implemented the Sun et al.~\cite{sun2011link}'s LS-join algorithm;
however, it did not return a solution even for the smallest instance (ca-GrQc).
The reason is that their algorithm is optimized to find the top-$k$ (with $k \le 100$) similar pairs between given two small sets.}:
\begin{enumerate}
  \item
    Yu et al.~\cite{yu2010taming}'s all-pairs SimRank algorithm.
    This algorithm computes all-pairs SimRank in $O(n m)$ time and $O(n^2)$ space.
    As discussed in \cite{yu2010taming}, we combine thresholding heuristics that discard small values (say $0.01$) in the SimRank matrix for each iteration.
    For the SimRank join, we first apply this algorithm and then output the similar pairs.

  \item
    Fogaras and Racz~\cite{fogaras2005scaling}'s random-walk based single-source SimRank algorithm.
    This algorithm computes single-source SimRank in $O(m)$ time and $O(n R)$ space, where $R$ is the number of Monte-Carlo samples.
    We set the number of Monte-Carlo samples $R = 1000$.
    For the SimRank join, we perform single-source SimRank computations for all seed vertices in parallel, and then output the similar pairs.
\end{enumerate}
We chose the parameters of these algorithms to hold similar accuracy (i.e., $\text{F-score} \approx 95\%$).
We used 7 datasets (3 small and 4 large datasets).
For small datasets, we also compute the exact SimRank scores and evaluate the accuracy.
For large datasets, we only compare the scalability.

The result is shown in Table~\ref{tbl:comparison};
each cell denotes the computational time, the allocated memory, and the F-score (for small datasets) or the number of similar pairs (for large datasets), respectively.
``---'' denotes the algorithm did not return a solution within 3 hour and 256 GB memory.

The results of small instances show that the proposed algorithm performs
the same level of accuracy to the existing algorithms with requiring much smaller memory and comparable time.
For large instances, the proposed algorithm outperforms the existing algorithms both in time and space,
because the complexity of the proposed algorithm depends on the number of similar pairs whereas the complexity of existing algorithms depends on the number of pairs, $O(n^2)$.
This shows the proposed algorithm is very scalable than the existing algorithms.

More precisely, the algorithm is efficient when
the number of similar pairs is relatively small.
For example, in email-Enron and web-Google datasets, which have many similar pairs,
the performance of the proposed algorithm close to
the existing algorithms.
On the other hand, in p2p-Gnutella31 and cit-patent dataset,
it clearly outperforms the existing algorithms.

\begin{table}[tb]
\tbl{Accuracy of the proposed algorithm.  \label{tbl:accuracy}}{
\centering
\tabcolsep=2pt
\begin{tabular}{l|ccccc} \hline
dataset & exact & obtained & precision & recall & F-score \\ \hline
as20000102 & 394026 & 377874 & 99\% & 95\% & 97\% \\
as-caida & 2727608 & 2601784 & 99\% & 95\% & 97\% \\
ca-CondMat & 6188 & 5964 & 96\% & 93\% & 95\% \\
ca-GrQc & 2707 & 2694 & 97\% & 96\% & 97\% \\
ca-HepTh & 4454 & 4340 & 99\% & 97\% & 98\% \\
p2p-Gnutella30 & 53752 & 49499 & 99\% & 92\% & 95\% \\
p2p-Gnutella31 & 98698 & 86809 & 99\% & 87\% & 93\% \\
wiki-Vote & 13199 & 12569 & 94\% & 90\% & 92\% \\
\hline
\end{tabular}

}
\end{table}

\begin{table}[tb]
\tbl{Comparison of algorithms.  \label{tbl:comparison}}{
\centering
\begin{tabular}{l|ccc} \hline
dataset & proposed & Fogaras\&Racz & Yu et al. \\ \hline
\multirow{3}{*}{as20000102}
& 10.0 s  & 18.4 s  & \textbf{7.4 s} \\
& \textbf{108 MB}  & 764 MB  & 289 MB \\ 
& 97\%  & 98\%  & \textbf{99\%} \\ 
\hline
\multirow{3}{*}{ca-GrQc}
& \textbf{1.5 s}   & 1.9 s & 1.9 s  \\
& \textbf{3 MB}  & 232 MB &  48 MB \\ 
& 97\%  & 97\%   & \textbf{99\%}  \\ 
\hline
\multirow{3}{*}{wiki-Vote}
& \textbf{0.4 s}   & 1.7 s & 12.1 s  \\
& \textbf{26.MB}  & 293 MB & 343 MB  \\ 
& \textbf{91\%} & \textbf{91\%}  & 88\%  \\ 
\hline
\multirow{3}{*}{p2p-Gnutella31}
& \textbf{2.7 s}   & 19.6 s  & 33.8 s \\
& \textbf{57 MB}   & 2.8 GB  & 1.2 GB \\ 
& 86,537   & 83,666   & 86,665 \\ 
\hline
\multirow{3}{*}{web-Google}
& \textbf{633.0 s}   & 1569.3 s  & 6537.6 s \\
& \textbf{24 GB}   & 36.6 GB  & 151.6 GB \\ 
& 11,414,971   & 11,444,009   & 10,980,706 \\ 
\hline
\multirow{3}{*}{soc-pokec}
& \textbf{182.9 s} & 2538.7 s  & --- \\
& \textbf{6 GB} & 76.6 GB  & --- \\ 
&    406,981 & 404,840  & --- \\ 
\hline
\multirow{3}{*}{cit-patent}
& \textbf{52.7 s}   & ---   & --- \\
& \textbf{3.6 GB}   & ---   & --- \\ 
& 209,900  & ---  & --- \\ 
\hline
\end{tabular}

}
\end{table}

\subsubsection{Experimental analysis of our algorithm}

In the previous subsections,
we observed that the proposed algorithm is scalable and accurate.
Moreover, it outperforms the existing algorithms.
In this subsection,
we experimentally evaluate the behavior of the algorithm.

We first evaluate the dependence with the accuracy parameter $\gamma$
used in the Gauss-Southwell algorithm.
We vary $\gamma$ and compute the lower set $J_L$ and the upper set $J_H$.
The result is shown in Table~\ref{tbl:gamma};
each cell denotes the computational time, the allocated memory, and the size of $J_L$ and $J_H$, respectively.

This shows, to obtain an accurate upper and lower sets,
we need to set $\gamma \ge 0.9$.
However, it requires $5$ times longer computational time
and $10$ times larger memory.
Since the verification phase can be performed in parallel,
to scale up for large instances,
it would be nice to set a small $\gamma$ (e.g., $\gamma = 0$).

\begin{table}[tb]
\tbl{Dependency of accuracy parameter $\gamma$.  \label{tbl:gamma}}{
\centering
\tabcolsep=2pt
\begin{tabular}{l|ccc} \hline
\multirow{2}{*}{dataset} & \multicolumn{3}{c}{$\gamma$} \\ 
& 0.0 & 0.5 & 0.9 \\ \hline
\multirow{3}{*}{as-caida}
& 6.8 s & 10.4 s & 33.6 s \\
& 777 MB & 1.1 GB & 2.0 GB \\ 
& [2.1M, 9.2M] & [2.2M, 5.7M] & [2.5M, 2.7M] \\
\hline
\multirow{3}{*}{com-amazon}
& 4.1 s & 9.0 s & 20.2 s \\
& 355 MB & 485 MB & 677 MB \\ 
& [96K, 2.1M] & [119K, 529K] & [144K, 178K] \\
\hline
\multirow{3}{*}{email-Enron}
& 4.0 s & 5.1 s & 15.8 s  \\
& 470 MB & 500 MB & 586 MB \\
& [1.2M, 2.9M] & [1.2M, 1.8M] & [1.2M, 1.4M] \\
\hline
\multirow{3}{*}{soc-Epinions1}
& 6.1 s & 6.9 s & 22.5 s \\
& 462 MB & 508 MB & 785 MB \\ 
& [205K, 1.2M] & [206K, 472K] & [208K, 291K] \\
\hline
\multirow{3}{*}{wiki-Vote}
& 0.3 s & 0.5 s & 1.8 s \\
& 26 MB & 27 MB & 33 MB \\ 
& [10K, 54K] & [10K, 22K] & [10K, 14K] \\
\hline
\end{tabular}

}
\end{table}

Next, we evaluate the dependence with the probability parameter $\beta$
in the stochastic thresholding.
We vary the parameter $\beta$ and evaluate the used memory.
The result is shown in Table~\ref{tbl:memory_reducing}.
The column for $\beta = \infty$ shows the result without this technique.
Thus we compare other columns with this column.
The result shows that the effect of memory-reducing technique
greatly depends on the network structure.
However, for an effective case, it reduces memory about $1/2$.
By setting $\beta = 100$, we can obtain almost the same result as
the result without the technique.

\begin{table}[tb]
\tbl{Dependency of the parameter $\beta$ for the stochastic thresholding.  \label{tbl:memory_reducing}}{
\centering
\tabcolsep=1pt
\begin{tabular}{l|cccc} \hline
\multirow{2}{*}{dataset} & \multicolumn{4}{c}{$\gamma$} \\
& $\infty$ & 1000 & 100 & 10 \\ \hline
\multirow{3}{*}{as-caida} & 7.0 s & 6.4 s & 5.8 s & 4.8 s \\ & 879 MB & 868 MB & 776 MB & 681 MB \\ & [2.1M,9.3M] & [2.1M,9.3M] & [2.1M,9.2M] & [2.1M,8.3M] \\ \hline
\multirow{3}{*}{com-amazon} & 4.8 s & 4.3 s & 3.9 s & 2.0 s \\ & 458 MB & 437 MB & 355 MB & 156 MB \\ & [97K,2.2M] & [97K,2.2M] & [96K,2.1M] & [90K,1.3M] \\ \hline
\multirow{3}{*}{email-Enron} & 4.7 s & 4.5 s & 3.6 s & 2.4 s \\ & 746 MB & 677 MB & 470 MB & 245 MB \\ & [1.2M,2.9M] & [1.2M,2.9M] & [1.2M,2.9M] & [1.2M,2.5M] \\ \hline
\multirow{3}{*}{soc-Epinions1} & 9.9 s & 8.8 s & 5.6 s & 3.2 s \\ & 1.1 GB & 954 MB & 462 MB & 135 MB \\ & [205K,1.2M] & [205K,1.2M] & [205K,1.2M] & [204K,968K] \\ \hline
\multirow{3}{*}{wiki-Vote} & 0.6 s & 0.5 s & 0.5 s & 0.5 s \\ & 68 MB & 56 MB & 26 MB & 7 MB \\ & [10K,54K] & [10K,54K] & [10K,54K] & [10K,43K] \\ \hline
\end{tabular}

}
\end{table}

Finally, we evaluate the effectiveness of the adaptive Monte-Carlo samples technique in the verification phase.
We plot the histogram of the number of required samples in Figure~\ref{fig:required_samples}. Here, the bar at 1,000 denotes the candidates that cannot be decided in 1,000 samples.

The result shows that most of small candidates are decided in 200 samples,
and $1/3$ of samples cannot decided in 1,000 samples.
Therefore, this implies the technique of
adaptive Monte-Carlo samples reduces
the computational cost in factor about $1/3$.

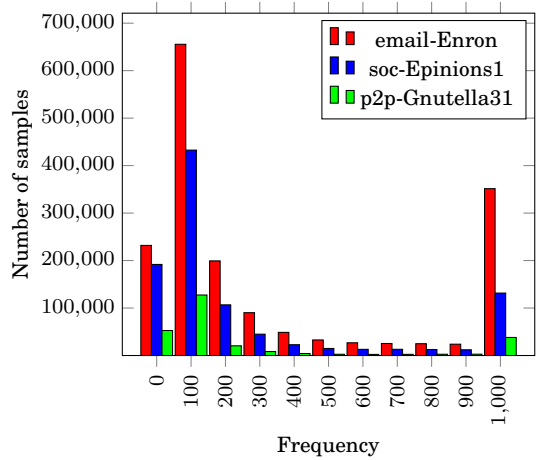
\begin{figure}[tb]
\centering
\begin{tikzpicture}
\footnotesize
\begin{axis}[
    scale=0.80,
    ybar=0pt,
    bar width=4,
    ymin=0,
    xlabel=Frequency,
    ylabel=Number of samples,
    ylabel near ticks,
    ytick = {100000,200000,300000,400000,500000,600000,700000},
    xtick = {0,100,200,300,400,500,600,700,800,900,1000},
    xticklabel style = {
      rotate=90,
      anchor=east,
      /pgf/number format/fixed,
    },
    scaled ticks=false,
    yticklabel style = {
      /pgf/number format/fixed,
    },
    y label style={at={(-0.2,0.5)}},
    x label style={at={(0.5,-0.1)}},
]
\addplot [ybar, fill=red] file {icde/result/required/email-Enron.required.txt};
\addlegendentry{email-Enron};
\addplot [ybar, fill=blue] file {icde/result/required/soc-Epinions1.required.txt};
\addlegendentry{soc-Epinions1};
\addplot [ybar, fill=green] file {icde/result/required/p2p-Gnutella31.required.txt};
\addlegendentry{p2p-Gnutella31};
\end{axis}
\end{tikzpicture}
\caption{Histogram of the number of required samples.}
\label{fig:required_samples}
\end{figure}

\subsubsection{Number of similar pairs}

Cai et al.~\cite{cai2009efficient} claimed that
the SimRank scores of a network follows power-law distributions.
However, they only verified this claim in small networks (at most 10K vertices).

We here verify this conjecture in larger networks.
We first enumerate the pairs with SimRank greater than $\theta = 0.001$,
and then compute the SimRank score by the Monte-Carlo algorithm.
The result is shown in Figure~\ref{fig:number_of_similar_pairs}.
This shows, for each experimented dataset,
the number of similar pairs follows power-law distributions in this range;
however, these exponent (i.e., the slope of each curve) are different.

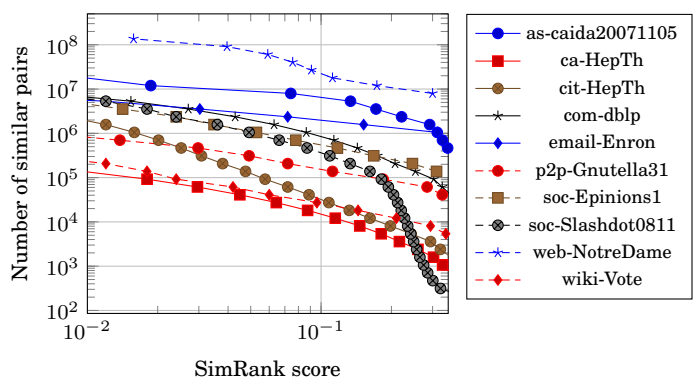
\begin{figure}[tb]
\begin{tikzpicture}
{
\footnotesize
\begin{loglogaxis}[
    xmin=0.01,
    xmax=0.35,
    xlabel=SimRank score,
    ylabel=Number of similar pairs,
    ylabel near ticks,
    grid,
    ytick={100,1000,10000,100000,1000000,10000000,100000000},
    legend style={at={(1.05,1.0)}, anchor=north west, font=\scriptsize},
    scale=0.70
  ]
\addplot file {icde/result/as-caida20071105.es.txt};
\addlegendentry{as-caida20071105};
\addplot file {icde/result/ca-HepTh.es.txt};
\addlegendentry{ca-HepTh};
\addplot file {icde/result/cit-HepTh.es.txt};
\addlegendentry{cit-HepTh};
\addplot file {icde/result/com-dblp.ungraph.es.txt};
\addlegendentry{com-dblp};
\addplot file {icde/result/email-Enron.es.txt};
\addlegendentry{email-Enron};
\addplot file {icde/result/p2p-Gnutella31.es.txt};
\addlegendentry{p2p-Gnutella31};
\addplot file {icde/result/soc-Epinions1.es.txt};
\addlegendentry{soc-Epinions1};
\addplot file {icde/result/soc-Slashdot0811.es.txt};
\addlegendentry{soc-Slashdot0811};
\addplot file {icde/result/web-NotreDame.es.txt};
\addlegendentry{web-NotreDame};
\addplot file {icde/result/wiki-Vote.es.txt};
\addlegendentry{wiki-Vote};
\end{loglogaxis}
}
\end{tikzpicture}
\caption{The number of similar pairs.}
\label{fig:number_of_similar_pairs}
\end{figure}



%

%
\bibliographystyle{ACM-Reference-Format-Journals}
\bibliography{main}

\section{Appendix}


In this appendix, we provide details of propositions and proofs mentioned in the main body of our paper.

We first introduce a vectorized form of SimRank, which is convenient for analysis.
Let $\otimes$ be the Kronecker product of matrices, i.e.,
for $n \times n$ matrices $A = (A_{ij})$ and $B$, 
$$ A \otimes B = \begin{bmatrix}
A_{11} B & \cdots & A_{1n} B \\
\vdots & \ddots & \vdots \\
A_{n1} B & \cdots & A_{nn} B
\end{bmatrix}. $$
Let ``$\mathrm{vec}$'' be the vectorization operator, which reshapes an $n \times n$ matrix to an $n^2$ vector, i.e., $ \mathrm{vec}(A)_{n \times i + j} = a_{i j} $.
Then we have the following relation, which is well known in linear algebra~\cite{AM05}:
\begin{align} \label{eq:VecKron}
\mathrm{vec}(A B C) = (C^\top \otimes A) \mathrm{vec}(B).
\end{align}

\begin{proposition}
\label{prop:NonSingular}
Linearized SimRank operator $S^L$ is a non-singular linear operator.
\end{proposition}
\begin{proof}
A linearized SimRank $S^L(\Theta)$ for a matrix $\Theta$ is a matrix satisfies relation
\[
  S^L(\Theta) = c P^\top S^L(\Theta) P + \Theta.
\]
By applying the vectorization operator and using \eqref{eq:VecKron}, we obtain
\begin{align} \label{eq:VectorizedSimRank}
  \left(I - c P^\top \otimes P^\top \right) \mathrm{vec}(S^L(\Theta)) = \mathrm{vec}(\Theta).
\end{align}
Thus, to prove Proposition~\ref{prop:NonSingular}, we only have to prove that the coefficient matrix
$ I - c P^\top \otimes P^\top $ is non-singular.

Since $P^\top \otimes P^\top$ is a (left) stochastic matrix,
its spectral radius is equal to one.
Hence all eigenvalues of $I - c P^\top \otimes P^\top$ are contained in the disk
with center $1$ and radius $c$ in the complex plane.
Therefore $I - c P^\top \otimes P^\top$ does not have a zero eigenvalue,
and hence $I - c P^\top \otimes P^\top$ is nonsingular.
\end{proof}

We now prove Proposition~\ref{prop:DiagonalCondition}, 
which is the basis of our diagonal estimation algorithm.
\begin{proof}[of Proposition~\ref{prop:DiagonalCondition}]
%
Let us consider linear system \eqref{eq:VectorizedSimRank}
and let $Q := P^\top \otimes P^\top$.
We partite the system \eqref{eq:VectorizedSimRank} into $2 \times 2$ blocks
that correspond to the diagonal entries and the others:
\begin{align} \label{eq:blockeq}
\begin{bmatrix} I - c Q_{DD} & - c Q_{DO} \\ -c Q_{OD} & I - c Q_{OO} \end{bmatrix}
\begin{bmatrix} \vec{1} \\ X \end{bmatrix}
=
\begin{bmatrix} \mathrm{diag}(D) \\ 0 \end{bmatrix},
\end{align}
where $Q_{DD}$, $Q_{DO}$, $Q_{OD}$, and $Q_{OO}$ are submatrices of $Q$ that denote contributions of diagonals to diagonals, diagonals to off-diagonals, off-diagonals to diagonals, and off-diagonals to off-diagonals, respectively.
$X$ is the off-diagonal entries of $S^L(D)$.
To prove Proposition~\ref{prop:DiagonalCondition}, we only have to prove that
there is a unique diagonal matrix $D$ that satisfies \eqref{eq:blockeq}.

We observe that the block-diagonal component $I - c Q_{OO}$ of \eqref{eq:blockeq} is non-singular,
which can be proved similarly as in Proposition~\ref{prop:NonSingular}.
Thus, the equation \eqref{eq:blockeq} is uniquely solved as
\begin{align}
\label{eq:closedformD}
  (I - c Q_{DD} - c^2 Q_{DO} (I - c Q_{OO})^{-1} Q_{OD}) \vec{1} = \mathrm{diag}(D),
\end{align}
which is a closed form solution of diagonal correction matrix $D$.
\end{proof}

\begin{remark}
\label{rem:initial}
It is hard to compute $D$ via the closed form formula \eqref{eq:closedformD}
because the evaluation of the third term of the left hand side of \eqref{eq:closedformD} is too expensive.

On the other hand, we can use \eqref{eq:closedformD} to obtain a reasonable initial solution
for Algorithm~\ref{alg:DiagonalEstimation}.
By using first two terms of \eqref{eq:closedformD}, we have
\begin{align} \label{eq:initial}
  \mathrm{diag}(D) \approx \vec{1} - c Q_{DD} \vec{1},
\end{align}
which can be computed in $O(m)$ time.

Our additional experiment shows that the initial solution \eqref{eq:initial} gives
a slightly better (at most twice) results than the trivial guesses $D = I$ and $D = 1 - c$.
\end{remark}

We give a convergence proof of our diagonal estimation algorithm (Algorithm~\ref{alg:DiagonalEstimation}).
As mentioned in Subsection~\ref{sec:Alternating}, 
Algorithm~\ref{alg:DiagonalEstimation} is the Gauss-Seidel method~\cite{golub2012matrix} for the linear system 
\begin{align}
\label{eq:iterationmatrix}
  \begin{bmatrix}
    S^L(E^{(1,1)})_{11} & \cdots & S^L(E^{(n,n)})_{11} \\
  \vdots & \ddots & \vdots \\
    S^L(E^{(1,1)})_{nn} & \cdots & S^L(E^{(n,n)})_{nn}
  \end{bmatrix}
  \begin{bmatrix}
  D_{11} \\ \vdots \\ D_{nn}
  \end{bmatrix}
  =
  \begin{bmatrix}
  1 \\ \vdots \\ 1
  \end{bmatrix}.
\end{align}

\begin{lemma}
\label{lem:DiagonallyDominant}
Consider two independent random walks start from the same vertex $i$
and follow their in-links.
Let $p_i(t)$ be the probability that two random walks meet $t$-th step (at some vertex).
Let $\Delta := \max_{i} \{ \sum_{t=1}^\infty c^t p_i(t) \}$.
If $\Delta < 1$ then the coefficient matrix of \eqref{eq:iterationmatrix} is diagonally dominant.
\end{lemma}
\begin{proof}
  By definition, each diagonal entry $S^L(E^{(j,j)})_{jj}$ is greater than or equal to one.
For the off diagonals, we have
\begin{align*}
  \sum_{i: i \neq j} S^L(E^{(i,i)})_{jj} &= \sum_{i:i \neq j} \sum_{t=1}^\infty c^t (P^t e_j)^\top E^{(i,i)} (P^t e_j) \\
  &\le \sum_{t=1}^\infty c^t (P^t e_j)^\top (P^t e_j) = \sum_{t=1}^\infty c^t p_j(t) \le \Delta.
\end{align*}
This shows that if $\Delta < 1$ then the matrix is diagonally dominant.
\end{proof}
\begin{corollary}
If a graph $G$ satisfies the condition of Lemma~\ref{lem:DiagonallyDominant}, 
Algorithm~\ref{alg:DiagonalEstimation} converges with convergence rate $O(\Delta^l)$.
\end{corollary}
\begin{proof}
This follows the standard theory of the Gauss-Seidel method~\cite{golub2012matrix}.
\end{proof}

\begin{remark}
Let us observe that, in practice, the assumption $\Delta < 1$ is not an issue.
For a network of average degree $d$, the probability $p_i(t)$ is expected to $1/d^t$.
Therefore $ \Delta = \sum_t c^t p_i(t) \simeq (c/d) / (1 - (c/d)) \le 1/(d-1) < 1$.
This implies that Algorithm~\ref{alg:DiagonalEstimation} converges quickly when the average degree is large.
\end{remark}

We now prove Proposition~\ref{prop:Hoeffding}. We use the following lemma.
\begin{lemma} \label{lem:lem1}
Let $k_1^{(t)}, \ldots, k_R^{(t)}$ be positions of $t$-th step of independent random walks that start from 
a vertex $k$ and follow ln-links.
Let $X_k^{(t)} := (1/R) \sum_{r=1}^R e_{k_r^{(t)}}$.
Then for all $l = 1, \ldots, n$,
\begin{align*}
  P\left\{ \left| e_l^\top \left( X_k^{(t)} - P^t e_k \right) \right| \ge \epsilon \right\} \le 2 \exp \left( - 2 R \epsilon^2 \right).
\end{align*}
\end{lemma}
\begin{proof}
Since $\E[ e_{k_r^{(t)}} ] = P^t e_k$,
this is a direct application of the Hoeffding's inequality.
\end{proof}
\begin{proof}[of Proposition~\ref{prop:Hoeffding}]
Since $p_{ki}^{(t)}$ defined by \eqref{eq:DefinitionPt} is represented by $p_{ki}^{(t)} = e_i^\top X^{(t)}_k$.
Thus we have
\begin{align*}
  P \left\{ \| P^t e_k - p^{(t)}_k \| > \epsilon \right\} 
  &\le
  n P \left\{ \left| e_i^\top P^t e_k - p^{(t)}_{ki} \right| > \epsilon \right\} \\
  &\le 
  2 n \exp \left( - 2 R \epsilon^2 \right). 
\end{align*}
\end{proof}

\begin{proposition}
\label{prop:UniformBound}
Let $D = \mathrm{diag}(D_{11}, \ldots, D_{nn})$ and 
$\tilde D = \mathrm{diag}(\tilde D_{11}, \ldots, \tilde D_{nn})$ be diagonal matrices.
If they satisfy \[
\sup_{k} |D_{kk} - \tilde D_{kk} | \le \epsilon
\]
then
\[
  \sup_{i, j} | S^L(D)_{ij} - S^L(\tilde D)_{ij} | \le \frac{\epsilon}{1 - c}.
\]
\end{proposition}
\begin{proof}
Let $\varDelta := D - \tilde D$. 
Since $S^L$ is linear, we have $S^L(\varDelta) = S^L(D) - S^L(\tilde D)$.
Consider
\[
  S^L(\varDelta) = c P^\top S^L(\varDelta) P + \varDelta.
\]
By applying $e_i$ and $e_j$, we have
\begin{align*}
  S^L(\varDelta)_{ij} &= c (P e_i)^\top S^L(\varDelta) (P e_j) + \varDelta_{ij} \\
                      &\le c \sup_{i',j'} S^L(\varDelta)_{i' j'} + \epsilon.
\end{align*}
Here, we used $p^\top A q \le \sup_{ij} A_{ij}$ for any stochastic vectors $p$ and $q$.
Therefore
\[
  (1 - c) \sup_{i,j} S^L(\varDelta)_{i j} \le \epsilon.
\]
By the same argument, we have
\[
  (1 - c) \inf_{i,j} S^L(\varDelta)_{i j} \ge -\epsilon.
\]
By combining them, we obtain the proposition.
\end{proof}
The above proposition shows that 
if diagonal correction matrix $D$ is accurately estimated,
all entries of SimRank matrix $S$ is accurately computed.

\renewcommand{\qed}{}

\medskip

\begin{proof}[of Proposition~\ref{prop:L1bound}]
Consider 
$ (P^t e_u)^\top D (P^t e_v)$.
Since $P^t e_v$ is a stochastic vector, by \eqref{eq:l1}, we have
\begin{align} \label{eq:tthbound}
  (P^t e_u)^\top D (P^t e_v) \le \max_{w \in \mathrm{supp}(P^t e_v)} (P^t e_u)^\top D e_w,
\end{align}
Since $P^t e_v$ corresponds to a $t$-step random walk,
the support is contained by a ball of radius $t$ centered at $v$.
Therefore we have
\begin{align*}
  (P^t e_u)^\top D (P^t e_v) \le \max_{w \in V : d-t \le d(u, w) \le d+t} (P^t e_u)^\top D e_w
\end{align*}
By plugging \eqref{eq:alpha}, we have
\[
  (P^t e_u)^\top D (P^t e_v) \le \max_{d - t \le d' \le d + t} \alpha(u, d', t).
\]
Substitute the above to \eqref{eq:ForSinglePair},
we obtain \eqref{eq:upperbound1}. 
\end{proof}
\begin{proof}[of Proposition~\ref{prop:L2bound}]
Consider 
$ c^t (P^t e_u)^\top D (P^t e_v)$.
By \eqref{eq:l2},  we have
\[
  (P^t e_u)^\top D (P^t e_v) = (\sqrt{D} P^t e_u)^\top  (\sqrt{D} P^t e_v) \le \gamma(u,t) \gamma(v,t).
\]
Substitute the above to \eqref{eq:ForSinglePair},
we obtain \eqref{eq:upperbound2}. 
\end{proof}

Let us start the proof of concentration bounds.
We first prepare some basic probablistic inequalities.
\begin{lemma}[Hoeffding's inequality]
  Let $X_1, \ldots, X_R$ be independent random variables with $X_r \in [0,1]$ for all $r = 1, \ldots, R$.
Let $S := (X_1 + \cdots + X_R)/R$.  Then
\[
  \P\left\{ \left| S - \E[S] \right| \ge \epsilon \right\} \le 2 \exp(-2 \epsilon^2 R).
\]
\end{lemma}
\begin{lemma}[Max-Hoeffding's inequality] \label{lem:hoeffding2}
For each $f = 1, \ldots, F$, let $X_r(f)$ ($r = 1, \ldots, R$) be independent random variables with $X_r(f) \in [0,1]$.
Let $S(f) := (X_1(f) + \cdots + X_R(f))/R$.  Then
\[
  \P\left\{ \max_f \left| S(f) - \E[S(f)] \right| \ge \epsilon \right\} \le 2 F \exp(-2 \epsilon^2 R).
\]
\end{lemma}
\begin{proof}
\begin{align*}
& \P\left\{ \max_f \left| S(f) - \E[S(f)] \right| \ge \epsilon \right\}  \\
&\le \sum_f \P\left\{ \left| S(f) - \E[S(f)] \right| \ge \epsilon \right\} \le 2 F \exp(-2 \epsilon^2 R). 
\end{align*}
\end{proof}

We write
$u_r^{(t)}$ and $v_r^{(t)}$ ($r = 1, \ldots, R$) for the $t$-th positions of independent random walks start from $u$ and $v$ and follow the in-links, respectively,
and $X_u^{(t)} := (1 / R) \sum_{r=1}^R e_{u_r^{(t)}}$, $X_v^{(t)} := (1 / R) \sum_{r=1}^R e_{v_r^{(t)}}$.

\begin{lemma} \label{lem:Dbound}
For each $w \in V$,
\begin{align*}
  \P\left\{ \left| X_u^{(t)\top} D e_w - (P^t e_u)^\top D e_w \right| \ge \epsilon \right\} \le 2 \exp(- 2 \epsilon^2 R).
\end{align*}
\end{lemma}
\begin{proof}
\begin{align*}
  & \P\left\{ \left| X_u^{(t) \top} D e_w - (P^t e_u)^\top D e_w \right| \ge \epsilon \right\} \\
  &\le \P\left\{ \left| X_u^{(t)^\top} e_w - (P^t e_u)^\top e_w \right| \ge \epsilon \right\}
  \le 2 \exp(-2 \epsilon^2 R). 
\end{align*}
\end{proof}

\begin{lemma} \label{lem:Duvbound}
\begin{align*}
  \P\left\{ \left| X_u^{(t) \top} D X_v^{(t)} - (P^t e_u)^\top D P^t e_v \right| \ge \epsilon \right\}
  \le 4 n \exp(- \epsilon^2 R / 2).
\end{align*}
\end{lemma}
\begin{proof}
\begin{align*}
& \P\left\{ \left| X_u^{(t) \top} D X_v^{(t)} - (P^t e_u)^\top D P^t e_v \right| \ge \epsilon \right\} \\
&\le \P \left\{ \left| X_u^{(t) \top} D \left( X_v^{(t)} - P^t e_v \right) \right| \ge \epsilon/2 \right\} \\
&+ \P \left\{ \left| \left( X_u^{(t)} - P^t e_u\right)^\top D P^t e_v \right| \ge \epsilon/2 \right\} \\
&\le 4 n \exp(-\epsilon^2 R / 2). 
\end{align*}
\end{proof}

\begin{proof}[of Proposition~\ref{prop:bound}]
By Lemma~\ref{lem:Duvbound}, we have
\begin{align*}
  & \P \left\{ \left| \left( \sum_{t=0}^{T-1} c^t X_u^{(t) \top} D X_v^{(t)} \right) - s^{(T)}(u,v) \right| \ge \epsilon \right\} \\
  \le & \sum_{t = 0}^{T-1} \P \left\{ \left| c^t X_u^{(t) \top} D X_v^{(t)} - c^t (P^t e_u)^\top D P^t e_v \right| \ge c^t \epsilon/(1-c)\right\} \\
  \le & 4 n T \exp \left( - \epsilon^2 R / 2(1-c)^2\right). 
\end{align*}
\end{proof}

\begin{proof}[of Proposition~\ref{prop:L1sample}]
We first prove the bound of $\alpha(u,d,t)$.
Note that
$\alpha(u,d,t) = \max_{w} \{ P^t e_u D e_w \}$
and the algorithm computes
$\tilde \alpha(u,d,t) = \max_{w} \{ X_u^{(t) \top} D e_w \}.$
By Lemmas~\ref{lem:Dbound} and \ref{lem:hoeffding2},
we have
\begin{align*}
&\P\left\{ \left| \max_w X_u^{(t) \top} D e_w - \max_w (P^t e_u)^\top D e_w \right| \ge \epsilon \right\} \\
&\le \P\left\{ \max_w \left| X_u^{(t) \top} D e_w - (P^t e_u)^\top D e_w \right| \ge \epsilon \right\} \\
& \le 2n \exp(-2 \epsilon^2 R).
\end{align*}
Using the above estimation, we bound $\beta$ as
\begin{align*}
  & \P\left\{ \left| \tilde \beta(u,d) - \beta(u,d) \right| \ge \epsilon \right\} \\
  & \le \sum_{d,t} \P\left\{ |\tilde \alpha(u,d,t) - \alpha(u,d,t)| \ge \epsilon \right\} \\
  & \le 2 n d_{\mathrm{max}} T \exp(- 2 \epsilon^2 R). 
\end{align*}
\end{proof}

%
\begin{proof}[of Proposition~\ref{prop:L2sample}]
We first observe that $\gamma(u,t)^2 = (P^t e_u)^\top D (P^t e_u)$
and the algorithm estimates this value by
\[
  (\tilde \gamma(u,t))^2 = \frac{1}{R^2} D_{u_r^{(t)} u_r^{(t)}}.
\]
Hence, by the same proof as Lemma~\ref{lem:Duvbound}, we have
\[
  \P\left\{ |\tilde \gamma(u,t)^2 - \gamma(u,t)^2| \ge \epsilon \right\} \le 4 n \exp(- \epsilon^2 R/2).
\]
Therefore
\begin{align*}
  & \P\left\{ | \tilde \gamma(u,t) - \gamma(u,t) | \ge \epsilon \right\} \\
  & \le \P\left\{ |\tilde \gamma(u,t)^2 - \gamma(u,t)^2| \ge \frac{\epsilon}{\tilde \gamma(u,t) + \gamma(u,t)} \right\} \\
  & \le \P\left\{ |\tilde \gamma(u,t)^2 - \gamma(u,t)^2| \ge \epsilon/2 \right\} \le 4 n \exp(- \epsilon^2 R / 8).
\end{align*}
Here we use the fact that both $\tilde \gamma(u,t)$ and $\gamma(u,t)$ are smaller than $\sqrt{ \max_w D_{ww} } = 1$. 
\end{proof}


\begin{proof}[Proof of Proposition~\ref{prop:terminate}]
We prove that Algorithm~\ref{alg:gausssouthwell} converges and also estimate the number of iterations.

Let $\langle A, B \rangle := \sum_{ij} A_{ij} B_{ij} = \mathrm{tr}(A^\top B)$ be the inner product of matrices.
Note that $\langle A B, C \rangle = \langle B, A^\top C \rangle = \langle A, C B^\top \rangle$.
We introduce a potential function of the form
\begin{align}
  \label{eq:potentialdef}
  \Phi(t) := \langle U, R^{(t)} \rangle,
\end{align}
where $U$ is a strictly positive matrix (determined later).
Since both $U$ and $R^{(t)}$ are nonnegative, 
the potential function $\Phi(t)$ is also nonnegative.
Moreover, $\Phi(t) = 0$ if and only if $R^{(t)} = O$ since $U$ is strictly positive.
To prove the convergence,
we prove that the potential function monotonically decreases.
More precisely, we prove that a strictly positive matrix $U$
exists for which the corresponding potential function decreases monotonically.

Let us observe that
\begin{align}
  \Phi(t+1) - \Phi(t) 
  &= \langle U, - R^{(t)}_{ij} E_{ij} + c R^{(t)}_{ij} P^\top E_{ij} P \rangle \nonumber \\ 
  &= -R^{(t)}_{ij} \langle U - c P U P^\top, E_{ij} \rangle.
\end{align}
Recall that, by the algorithm, $R^{(t)}_{ij} > \epsilon$. 
Thus, to guarantee $\Phi(t+1) - \Phi(t) < 0$,
it suffices to prove the existence of matrix $U$ satisfying
\begin{align}
  \label{eq:conditionU}
  U > 0, \quad U - c P U P^\top > O,
\end{align}
where $U > O$ denotes that all entries of the matrix $U$ is strictly positive.
We explicitly construct this matrix.
Let $E$ be the all-one matrix. 
Then the matrix
\begin{align} \label{eq:definitionU}
  U := E + c P E P^\top + c^2 P^2 E P^{\top 2} + \cdots.
\end{align}
satisfies \eqref{eq:conditionU} as follows.
First, $U$ is strictly positive 
because the first term in \eqref{eq:definitionU} is strictly positive and the other terms are nonnegative.
Second, since
\begin{align}
  U - c P U P^\top = E,
\end{align}
it is also strictly positive.
Therefore, using this matrix $U$ in \eqref{eq:potentialdef},
we can obtain that the potential function $\Phi(t)$ is nonnegative and is strictly monotonically decreasing.
This proves the convergence of $\Phi(t) \to 0$.
Furthermore, since $\Phi(t) = 0$ implies $R^{(t)} = O$, 
this proves the convergence of $R^{(t)} \to O$.

Let us bound the number of iterations.
From the above analysis, we obtain
\begin{align}
  \Phi(t+1) - \Phi(t) = - R^{(t)}_{ij} < -\epsilon.
\end{align}
Thus the number of iterations is bounded by $O(\Phi(0) / \epsilon)$.
The rest of the proof, we bound $\Phi(0)$.
For the initial solution $R^{(0)} = D$, we have 
\begin{align}
  \langle U, R^{(0)} \rangle 
  &= \langle U, D \rangle 
  = \langle \sum_{t=0}^\infty c^t P^t E P^{\top t}, D \rangle \nonumber \\
  &= \langle E, \sum_{t=0}^\infty c^t P^{\top t} D P^t \rangle 
  = \langle E, S \rangle = \sum_{ij} S_{ij}.
\end{align}
For the third equality, we used 
\begin{align}
  S = D + c P^\top D P + \cdots,
\end{align}
which follows from \eqref{eq:linearizedsimrank}.
This shows $\Phi(0) = \sum_{ij} S_{ij}$.
\end{proof}


\begin{proof}[Proof of Proposition~\ref{prop:errorestimate}]
When the algorithm terminates, we obtain a solution $\tilde S$ with residual $\tilde R$ satisfying 
the following bound:
\[
  \tilde R_{ij} = \left( D - (\tilde S - c P^\top \tilde S P) \right)_{ij} \le \epsilon.
\]
Recall that $\tilde R_{ij} \ge 0$ by construction.
We establish an error bound for the solution $\tilde S$ from the above bound of the residual $\tilde R$.
Recall also that the SimRank matrix satisfies $S - c P^\top S P = D$.
Thus we have
\begin{align*}
  \left( (S - \tilde S) - c P^\top (S - \tilde S) P \right)_{ij} \le \epsilon.
\end{align*}
We can evaluate the second term as
\begin{align*}
  \left( P^\top (S - \tilde S) P \right)_{ij} \le \max_{ij} (S_{ij} - \tilde S_{ij}).
\end{align*}
Therefore we obtain
\begin{align*}
  \max_{ij} (S_{ij} - \tilde S_{ij}) \le \frac{\epsilon}{1 - c}.
\end{align*}
\end{proof}

\begin{proof}[Proof of Proposition~\ref{prop:thresholding}]
Let $k$ be the smallest index such that $a_1 + \cdots + a_k > \delta$.
Then the error, $A - \tilde A$, exceeds $\delta$ if and only if the first $k$ values are skipped.
If $\beta a_i \ge 1$ for some $1 \le i \le k$, it must not be skipped, therefore the proposition holds.
Otherwise, the probability is given as
\begin{align}
  \P\{ A - \tilde A \ge \delta \} = (1 - \beta a_1) \cdots (1 - \beta a_k).
\end{align}
By the arithmetic mean-geometric mean inequality, we have
\begin{align*}
  (1 - \beta a_1) \cdots (1 - \beta a_k) &\le \left(1 - \frac{1}{k} \sum_{i=1}^k \beta a_i\right)^k  \\
                                   &\le \left(1 - \frac{\beta \delta}{k} \right)^k \le \exp(-\beta \delta).
\end{align*}
Therefore the proposition holds.
\end{proof}

\begin{proof}[Proof of Proposition~\ref{prop:memory}]
For each iteration, we have the following invariant:
\begin{align}
  D - \left( S^{(t)} - c P^\top S^{(t)} P\right) = R^{(t)} + \bar R^{(t)},
\end{align}
where $\tilde R^{(t)}$ is the skipped values by the stochastic thresholding. 
Using this invariant, 
this proposition follows from the similar proof as Proposition~\ref{prop:errorestimate} with Proposition~\ref{prop:thresholding}.
\end{proof}

\begin{proof}[Proof of Proposition~\ref{prop:hoeffding}]
Since $\E[ s^{(R)}(i,j) ] = s(i,j)$, 
by the Hoeffding inequality, we have
\begin{align}
  \P\{ s^{(R)}(i,j) \ge s(i,j) + \epsilon \} \le \exp(-2 R \epsilon^2).
\end{align}
Therefore
\begin{align}
  \P\{ s^{(R)}(i,j) \ge \theta \} \le \exp(-2 R \epsilon^2).
\end{align}
\end{proof}

\end{document}